\documentclass[journal]{IEEEtran}
\usepackage{amsfonts}
\usepackage{multirow,array,enumerate,color,graphicx,fancybox,pifont,epsf,epsfig,subfigure,amsmath,amssymb,psfrag}
\usepackage{algorithm}
\usepackage{algorithmic}
\usepackage{cite}
\newcounter{MYtempeqncnt}
\newtheorem{prob}{\textbf{Problem}}

\newtheorem{theorem}{\textbf{Theorem}}

\newtheorem{assumption}[theorem]{\textbf{Assumption}}

\newcommand{\secref}[1]{Section~\ref{#1}}

\newcommand{\figref}[1]{Figure~\ref{#1}}
\newcommand{\assref}[1]{Assumption~\ref{#1}}

\newcommand{\theoref}[1]{Theorem~\ref{#1}}

\newcommand{\proref}[1]{Proposition~\ref{#1}}

\newcommand{\remref}[1]{Remark~\ref{#1}}

\newcommand{\algref}[1]{Algorithm~\ref{#1}}

\newtheorem{theo}{\textbf{Theorem}}
\newtheorem{pro}{\textbf{Proposition}}

\newtheorem{rem}{\textbf{Remark}}

\renewcommand{\QED}{\hfill\blacksquare}

\newcommand{\qed}{\nobreak \ifvmode \relax \else
      \ifdim\lastskip<1.5em \hskip-\lastskip
      \hskip1.5em plus0em minus0.5em \fi \nobreak
      \vrule height0.75em width0.5em depth0.25em\fi}

\title{Distributed Quantization for Measurement of Correlated Sparse Sources over Noisy Channels\thanks{This work is partially presented in International Conference on Acoustics, Speech, and Signal Processing (ICASSP), Florence, Italy, May 2014.}}
\author{Amirpasha Shirazinia, \textit{Student Member, IEEE}, Saikat Chatterjee, \textit{Member, IEEE}, Mikael Skoglund, \textit{Senior Member, IEEE} \\
}

\begin{document}
\maketitle

\begin{abstract}
In this paper, we design and analyze distributed vector quantization (VQ) for compressed measurements of correlated sparse sources over noisy channels. Inspired by the framework of compressed sensing (CS) for acquiring compressed measurements of the sparse sources, we develop optimized quantization schemes that enable distributed encoding and transmission of CS measurements over noisy channels followed by joint decoding at a decoder. The optimality is addressed with respect to minimizing the sum of mean-square error (MSE) distortions between the sparse sources and their reconstruction vectors at the decoder. We propose a VQ encoder-decoder design via an iterative algorithm, and derive a lower-bound on the end-to-end MSE of the studied distributed system. Through several simulation studies, we evaluate the performance of the proposed distributed scheme.
\end{abstract}

\begin{IEEEkeywords}
   \noindent  Vector quantization, distributed compression, correlation, sparsity, compressed sensing, noisy channel.
\end{IEEEkeywords}

\section{Introduction} \label{sec:intro}
Source compression is one of the most important and contributing factors in developing digital signal processing. Various source compression approaches can be combined together in order to realize a better source compression scheme. In this paper, we endeavour to combine the strength of two standard compression approaches: (1) vector quantization (VQ) \cite{98:Gray} and its extension to transmission over noisy channels, and (2) compressed sensing (CS) \cite{08:Candes} -- a linear dimensionality reduction framework for sources that can be represented by sparse structures. We use VQ since it is theoretically the optimal block (vector) coding strategy \cite{98:Gray}. This is because of space-filling advantage (corresponding to dimensionality), shaping advantage (corresponding to probability density function) and memory advantage (corresponding to correlations between components) of VQ \cite{89:Lookabaugh} over structured quantizers, such as scalar or uniform quantizers. On the other hand, inspired by the CS framework, it is guaranteed to acquire few measurements from a sparse-structured signal vector without losing useful information, and to accurately reconstruct the original signal. We employ the VQ and CS compression approaches within a distributed setup, with correlated sparse sources, for transmission over noisy channels. Distributed source compression approaches (see, e.g., \cite{73:Slepian,76:Wyner-Ziv,99:Zamir,04:Zixiang,04:Chen,05:Oohama,06:Gastpar,06:Rebollo,08:Wagner,09:Dragotti,09:Wernersson,09:Niklas,13:Sun,14:Xuechen}) are of high practical relevance, and in modern applications, multiple remote sensors may observe a physical phenomenon. As a consequence, they are not able to cooperate with each other, and need to accomplish their tasks independently.

So far in literature, there is no attempt to investigate a unified scenario where a distributed channel-robust VQ scheme is applied for compressed sensing of correlated sparse sources. We attempt to deal with such a unified scenario via developing new algorithms and theory. Without loss of generality, we consider two correlated sparse sources. Each source is independently measured via a CS-based sensor. Then each of the two measurement vectors is independently quantized via a channel-robust VQ scheme. Finally, at the decoder, both sources are jointly reconstructed. For such a distributed setup, natural questions are: (1) How to design VQ for CS measurements that is robust against channel noise? (2) What is the theoretical performance limit of such system? We endeavour to answer both questions in this paper.

In a CS setup, quantization of CS measurement vector is an important issue due to the requirement of finite bit digital representation. Attempts have been made in literature to bring quantization and compressed sensing together,
but neither to use a distributed quantization setup nor to address robustness of quantizer when transmissions are made over noisy channels.
Some examples of existing quantization schemes for compressed sensing are as follows. In \cite{06:Candes2,10:Sinan,10:Zymnis,11:Dai,11:Jacques,12:Yan,12:Kamilov},
new CS reconstruction schemes have been developed in order to mitigate the effect of quantization. On the other hand \cite{09:Sun,12:Boufounos,11:Kamilov,12:Pasha1} considered
development of new quantization schemes to suit a CS reconstruction algorithm. Considering the aspect of non-linearity in any standard CS reconstruction, we recently developed analysis-by-synthesis-based quantizer in \cite{13:Pasha_journal}. Also, the work of \cite{08:Goyal,11:Dai,12:Laska,12:Pasha1} addressed the trade-off between resources of quantization (quantization bit rate) and CS (number of measurements). Further, \cite{07:Bajwa,09:Baron,12:Rambeloarison}
considered distributed CS setups, but without any quantization. Some works have studied connection between network coding and CS \cite{10:Feizi,12:Nabaee}, and between distributed lossless coding and CS \cite{09:Cheng}.

\subsection{Contributions}
We consider a distributed setup comprising two CS-based sensors measuring two correlated sparse source vectors.
The low-dimensional, and possibly noisy measurements are quantized using a VQ, and transmitted over
independent discrete memoryless channels (DMC's). The sparse source vectors are reconstructed
at the decoder from received noisy symbols. We use sum of mean square error (MSE) distortions
between the sparse source vectors and their reconstruction vectors at the decoder as the performance
criterion. The performance measure corresponds to the end-to-end MSE which will be described later.
Our contributions are as follows:
\begin{itemize}
    \item Establishing (necessary) conditions for optimality of VQ encoder-decoder pairs.
    \item Developing a VQ encoder-decoder design algorithm through an iterative algorithm.
    \item Deriving a lower-bound on the MSE performance.
\end{itemize}

For optimality of the VQ encoder-decoder pairs, we minimize the end-to-end MSE, and require to use the Bayesian framework of minimum mean square error (MMSE) estimation. Hence, We do not use prevalent CS reconstruction algorithms. We illustrate the performance of the proposed distributed design via simulation studies by varying correlation, compression resources and channel noise, and compare it with the derived lower-bound and centralized schemes.

\subsection{Outline}

The rest of the paper is organized as follows. In \secref{sec:descrpn}, we describe a two-sensor  distributed system model that we study; the description involves building blocks, performance criterion and objectives. \secref{sec:Design} is devoted to preliminaries and design of encoder-decoder pairs in a distributed fashion. Preliminaries, in \secref{sec:pre analysis}, include developing optimal estimation of correlated sparse sources from noisy CS measurements which helps us to design optimized encoding schemes, in \secref{subsec:enc design}, and decoding schemes, in \secref{subsec:dec design}. Thereafter, in \secref{sec:training}, we develop an encoder-decoder training algorithm. The end-to-end performance analysis of the studied distributed system is given in \secref{sec:analysis}. The performance evaluation is made in \secref{sec:numerical}, and the conclusions are drawn in \secref{sec:conclusion}.

\textit{Notations:} Random variables (RV's) will be denoted by upper-case letters while their realizations (instants) will be denoted by the respective lower-case letters. Hence, if $\mathbf{Z}$ denotes a random row vector $[Z_1,\ldots,Z_n]$, then $\mathbf{z} = [z_1,\ldots,z_n]$ indicates a  realization of $\mathbf{Z}$. Matrices will be represented by boldface characters. The trace of a matrix is shown by $\text{Tr}\{\cdot\}$ and transpose of a vector/matrix by $(\cdot)^\top$. Further, cardinality of a set is shown by $|\cdot|$.  We will use~$\mathbb{E}[\cdot]$ to denote the expectation operator, and conditional expectation $\mathbb{E}[Z|y]$ indicates  $\mathbb{E}[Z|Y=y]$. The $\ell_p$-norm ($p > 0$) of a vector $\mathbf{z}$ will be denoted by $\|\mathbf{z}\|_p = (\sum_{n=1}^N |z_n|^p)^{1/p}$. Also, $\|\mathbf{z}\|_0$ represents $\ell_0$-norm which is the number of non-zero coefficients in $\mathbf{z}$. 

\section{System Description and Problem Statement} \label{sec:descrpn}
In this section, we describe the system, depicted in \figref{fig:diagram_dist}, and associated assumptions. 

\begin{figure}[t]
  \begin{center}
  \psfrag{x_1}[][][0.7]{$\mathbf{X}_1$}
  \psfrag{x_2}[][][0.7]{$\mathbf{X}_2$}
  \psfrag{A_1}[][][1]{$\mathbf{\Phi}_1$}
  \psfrag{A_2}[][][1]{$\mathbf{\Phi}_2$}
  \psfrag{y_1}[][][0.7]{$\mathbf{Y}_1$}
  \psfrag{y_2}[][][0.7]{$\mathbf{Y}_2$}
  \psfrag{w_1}[][][0.6]{$\mathbf{W}_1$}
  \psfrag{w_2}[][][0.6]{$\mathbf{W}_2$}
  \psfrag{Q_1}[][][0.85]{$\textsf{E}_1$}
  \psfrag{Q_2}[][][0.85]{$\textsf{E}_2$}
  \psfrag{i_1}[][][0.7]{$I_1$}
  \psfrag{i_2}[][][0.7]{$I_2$}
  \psfrag{j_1}[][][0.7]{$J_1$}
  \psfrag{j_2}[][][0.7]{$J_2$}
  \psfrag{DMC_1}[][][0.75]{$P(j_1|i_1)$}
  \psfrag{DMC_2}[][][0.75]{$P(j_2|i_2)$}
  \psfrag{Channel}[][][0.75]{Channel}
  \psfrag{quant}[][][0.75]{Quantizer}
  \psfrag{Compressed}[][][0.75]{Compressed}
  \psfrag{Sensing}[][][0.75]{sensing}
  \psfrag{Enc}[][][0.75]{Encoder}
  \psfrag{CS Enc}[][][0.75]{CS}
  \psfrag{Dec}[][][0.75]{Decoder}
  \psfrag{D_1}[][][0.85]{$\textsf{D}_1$}
  \psfrag{D_2}[][][0.85]{$\textsf{D}_2$}
  \psfrag{x_h_1}[][][0.7]{$\widehat{\mathbf{X}}_1$}
  \psfrag{x_h_2}[][][0.7]{$\widehat{\mathbf{X}}_2$}
  \includegraphics[width=9cm]{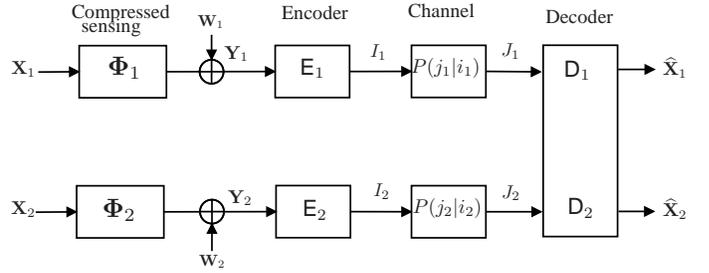}\\
  \caption{Distributed vector quantization for CS measurements over noisy channels.}\label{fig:diagram_dist}
  \end{center}
\end{figure}

\subsection{Compressed Sensing, Encoding, Transmission Through Noisy Channel and Decoding} \label{subsec:CS system}
We consider a $K$-sparse (in a known basis) vector $\mathbf{\Theta} \in \mathbb{R}^N$ comprised of $K$ random non-zero coefficients ($K \ll N$). We define the support set, i.e., random location of non-zero coefficients, of the vector $\mathbf{\Theta} \triangleq [\Theta_{1},\ldots,\Theta_N]^\top$ as $\mathcal{S} \triangleq \{n \in \{1,2,\ldots,N\}: \Theta_n \neq 0 \}$ with $|\mathcal{S}| = \|\mathbf{\Theta}\|_0 = K$. Further, we assume two correlated sparse sources $\mathbf{X}_1 \in \mathbb{R}^N$ and $\mathbf{X}_2 \in \mathbb{R}^N$ have a common support set in which their correlation is established by the following model
\begin{equation} \label{eq:correlation}
    \mathbf{X}_l = \mathbf{\Theta} + \mathbf{Z}_l, \hspace{0.2cm} l \in \{1,2\},
\end{equation}
where $\mathbf{Z}_l \triangleq [Z_{1,l}, \ldots ,Z_{N,l}]^\top \in \mathbb{R}^N$ is a random $K$-sparse vector with a common support set $\mathcal{S}$; thus $\|\mathbf{Z}_l\|_0 = K$. We also assume that $\mathbf{Z}_1$ and $\mathbf{Z}_2$ are uncorrelated with each other and with the common signal vector $\mathbf{\Theta}$. Such a joint sparsity model (JSM), also known as JSM-2, was earlier used for distributed CS in \cite{09:Baron}. Interested readers are referred to  \cite{05:Gilbert,05:Tropp,09:Baron} for application examples of JSM-2. 

The correlated sparse sources $\mathbf{X}_1$ and $\mathbf{X}_2$ are measured by CS-based sensors, leading to measurement vectors $\mathbf{Y}_1 \in \mathbb{R}^{M_1}$ and $\mathbf{Y}_2 \in \mathbb{R}^{M_2}$ described by equations
\begin{equation} \label{eq:meas eq1}
    \mathbf{Y}_l = \mathbf{\Phi}_l \mathbf{X}_l + \mathbf{W}_l, \hspace{0.2cm} l \in \{1,2\}, \hspace{0.1cm} \|\mathbf{X}_l\|_0 = K,
\end{equation}
where $\mathbf{\Phi}_l \in  \mathbb{R}^{M_l \times N}$ is a fixed sensing matrix of the $l^{th}$ sensor, and there is no specific model is assumed on the sensing matrix. Further, $\mathbf{W}_l \in \mathbb{R}^{M_l}$ is an additive measurement noise vector independent of other sources. Without loss of generality, we will assume that $M_1 = M_2 \triangleq M$, and according to CS requirement $M  < N$.

The encoders at the terminals have access to the correlated sparse sources indirectly through the noisy and lower-dimensional CS measurements. The encoder at terminal $l$ ($l \in \{1,2\}$) codes the noisy CS measurement vector $\mathbf{Y}_l$ without cooperation with the other encoder. The encoder mapping $\textsf{E}_l$  encodes $\mathbf{Y}_l$ to a transmission index $i_l$, i.e.,
\begin{equation} \label{eq:encoders}
    \textsf{E}_l: \mathbb{R}^M \rightarrow \mathcal{I}_l, \hspace{0.5cm}  l \in \{1,2\},
\end{equation}
where $i_l \in \mathcal{I}_l$, and $\mathcal{I}_l$ denotes a finite index set defined as $\mathcal{I}_l \triangleq \{0,1,\! \ldots \!,2^{R_l}-1\}$ with $|\mathcal{I}_l| \! \triangleq \!\mathfrak{R}_l\! =\! 2^{R_l}$. Here, $R_l$ is the assigned quantization rate for the $l^{th}$ encoder in bits/vector. We fix the total quantization rate at $R_1 \!+\! R_2 \triangleq R$ bits/vector. The encoders are specified by the regions $\{\mathcal{R}_{i_l}\}_{i_l=0}^{\mathfrak{R}_l-1}$ where $\bigcup_{i_l=0}^{\mathfrak{R}_l-1} \mathcal{R}_{i_l} \!=\! \mathbb{R}^M$ such that when $\mathbf{Y}_l \! \in \! \mathcal{R}_{i_l}$, the encoder outputs $\textsf{E}_l(\mathbf{Y}_l)\! =\! i_l \in \mathcal{I}_l$.

For transmission, we consider discrete memoryless channels (DMC's) consisting of discrete input and output alphabets, and transition probabilities. 
The DMC's accept the encoded indexes $i_l$, and output noisy symbols $j_l \in \mathcal{I}_l$, $l \in \{1,2\}$. The channel is defined by a random mapping $\mathcal{I}_l \rightarrow  \mathcal{I}_l$ characterized by known transition probabilities
\begin{equation} \label{eq:channel trans}
P(j_l | i_l) \triangleq \textrm{Pr}(J_l = j_l | I_l = i_l), \hspace{0.15cm} i_l,j_l \in \mathcal{I}_l, \forall l \in \{1,2\}.
\end{equation}

Finally, each decoder uses both noisy indexes $j_1 \in \mathcal{I}_1$ and $j_2 \in \mathcal{I}_2$ in order to make the estimate  of the sparse source vector, denoted by $\widehat{\mathbf{X}}_l \in \mathbb{R}^N$, $l \in \{1,2\}$. Given the received indexes $j_1$ and $j_2$, the decoder $\textsf{D}_l$ is characterized by a mapping
\begin{equation} \label{eq:decoders}
 \textsf{D}_l: \mathcal{I}_1 \times \mathcal{I}_2 \rightarrow \mathcal{C}_l, \hspace{0.2cm} l \in \{1,2\},
\end{equation}
where $\mathcal{C}_l \subseteq \mathbb{R}^N \times \mathbb{R}^N$, with $|\mathcal{C}_l| = 2^{R_1 + R_2}$, is a finite discrete \textit{codebook} set containing all reproduction \textit{codevectors}. The decoder's functionality is described by a look-up table; $(J_1 = j_1, J_2 = j_2) \Rightarrow (\widehat{\mathbf{X}}_1 = \textsf{D}_1(j_1,j_2),\widehat{\mathbf{X}}_2 = \textsf{D}_2(j_1,j_2))$.

\subsection{Performance Criterion} \label{sec:criterion}
We use end-to-end MSE as the performance criterion, defined as
\begin{equation} \label{eq:MSE}
    D \triangleq \frac{1}{2K} \sum_{l=1}^2 \mathbb{E}[\|\mathbf{X}_l - \widehat{\mathbf{X}}_l \|_2^2].
\end{equation}
Note that the MSE depends on \textit{CS reconstruction distortion}, \textit{quantization error} as well as \textit{channel noise}. Our goal, stated below, is to design VQ encoder-decoder pairs robust against all these three kinds of error.
\begin{prob}
Consider the system of \figref{fig:diagram_dist} for distributed VQ of CS measurements over DMC's. Given fixed quantization rates $R_l$ ($l \in \{1,2\}$) at terminal $l$, known sensing matrices $\mathbf{\Phi}_l$, and channel transition probabilities $P(j_l|i_l)$, we aim to find
\begin{itemize}
    \item encoder mapping $\textsf{E}_l$ in \eqref{eq:encoders} to separately encode CS measurements, and
    \item decoder mapping $\textsf{D}_l$ in \eqref{eq:decoders} to jointly decode correlated sparse sources,
\end{itemize}
such that the end-to-end MSE, in \eqref{eq:MSE}, is minimized.
\end{prob}

\section{Design Methodology} \label{sec:Design}
In this section, we show how to optimize the encoder and decoder mappings of the system of \figref{fig:diagram_dist}. We are aware that a fully joint design of the encoder and decoder mappings is intractable. Therefore, we optimize each mapping (with respect to minimizing the MSE in \eqref{eq:MSE}) by fixing the other mappings. Therefore, the resulting mappings fulfil necessary conditions for optimality. We first begin with some analytical preliminaries.

\subsection{Preliminaries} \label{sec:pre analysis}
Before proceeding with the design methodology to obtain the optimized encoder-decoder pair in order to minimize MSE, we need to develop some analytical results, discussed below. We first mention our assumptions.

\begin{assumption} \label{ass1}
    $\\$
\begin{enumerate}
\vspace{-0.5cm}
    \item The elements of the support set $\mathcal{S}$ are drawn uniformly at random from the set of all ${N \choose K}$ possibilities, denoted by $\mathbf{\Omega}$.
    \item The non-zero coefficients of $\mathbf{\Theta}$ and $\mathbf{Z}_l$ ($l \in \{1,2\}$) are iid Gaussian RV's with zero mean and variance $\sigma_\theta^2$ and $\sigma_{z_l}^2$, respectively. Without loss of generality, we assume that $\sigma_{z_1}^2 = \sigma_{z_2}^2 \triangleq \sigma_z^2$ and $\sigma_\theta^2 + \sigma_z^2 = 1$, i.e., the variance of a non-zero component in $\mathbf{X}_l$ is normalized to $1$.
    \item The measurement noise vector is distributed as $\mathbf{W}_l \sim \mathcal{N}(\mathbf{0}, \sigma_{w_l}^2 \mathbf{I}_M)$, $l \in \{1,2\}$, which is uncorrelated with the CS measurements and sources.
\end{enumerate}
\end{assumption}

To measure the amount of correlation between sources, we define the \textit{correlation ratio} as
\begin{equation} \label{eq:corr ratio}
    \rho \triangleq \sigma_\theta^2 / \sigma_z^2.
\end{equation}
Hence, $\sigma_\theta^2 = \frac{\rho}{1+\rho}$ and $\sigma_z^2 = \frac{1}{1+\rho}$, and $\rho \! \rightarrow \! \infty$ implies that the sources are highly correlated, whereas $\rho \! \rightarrow \! 0$ means that they are highly uncorrelated. Next, we define reconstruction distortion of the sparse sources from noisy CS measurements, termed \textit{CS distortion}, as
\begin{equation} \label{eq:CS distotion}
    D_{cs} \triangleq \frac{1}{2K} \sum_{l=1}^2 \mathbb{E}[\|\mathbf{X}_l - \widetilde{\mathbf{X}}_l\|_2^2],
\end{equation}
where $\widetilde{\mathbf{X}}_l \in \mathbb{R}^N$ ($l \in \{1,2\}$) is an estimation vector of the sparse source $\mathbf{X}_l$ from noisy CS measurements $\mathbf{Y}_1$ and $\mathbf{Y}_2$. Further, to minimize $D_{cs}$ in \eqref{eq:CS distotion}, we need to derive MMSE estimator of correlated sparse sources given noisy CS measurements. The following proposition provides an analytical expression for the MMSE estimator, which is also useful in deriving bounds later on the CS distortion (in \proref{pro:oracle bound}) and end-to-end distortion (in \theoref{theo2}).

\begin{pro}[\textit{MMSE estimation}] \label{theo1}
    Consider the linear noisy CS measurement equations in \eqref{eq:meas eq1} under \assref{ass1}. Then, the MMSE estimation of $\mathbf{X}_l$ given the noisy CS measurement vector $\mathbf{y} \triangleq [\mathbf{y}_1^\top \hspace{0.15cm} \mathbf{y}_2^\top]^\top$ that minimizes $D_{cs}$ in \eqref{eq:CS distotion}, is obtained as $\widetilde{\mathbf{x}}_l^\star(\mathbf{y}) = \mathbb{E}[\mathbf{X}_l|\mathbf{y}]$ which has the following closed form expression
    \begin{equation} \label{eq:MMSE closed}
    \begin{aligned}
        &\widetilde{\mathbf{x}}^\star(\mathbf{y}) \triangleq [\widetilde{\mathbf{x}}_1^\star(\mathbf{y})^\top \hspace{0.1cm} \widetilde{\mathbf{x}}_2^\star(\mathbf{y})^\top]^\top  = \frac{\sum_{\mathcal{S} \subset \mathbf{\Omega}} \beta_\mathcal{S} \cdot \widetilde{\mathbf{x}}^\star(\mathbf{y},\mathcal{S})}{\sum_{\mathcal{S} \subset \mathbf{\Omega}} \beta_\mathcal{S}}, &
    \end{aligned}
    \end{equation}
    where $\widetilde{\mathbf{x}}^\star(\mathbf{y},\mathcal{S}) \triangleq \mathbb{E}[\mathbf{X}|\mathbf{y},\mathcal{S}]$, in which $\mathbf{X} \triangleq [\mathbf{X}_1^\top \hspace{0.15cm} \mathbf{X}_2^\top]^\top$, and within its support
    \begin{equation} \label{eq:MMSE closed details1}
    \begin{aligned}
        &\widetilde{\mathbf{x}}^\star(\mathbf{y},\mathcal{S}) =  \left[\begin{array}{c  c  c}
        \mathbf{I}_K & \mathbf{I}_K & \mathbf{0}_K  \\
        \mathbf{I}_K & \mathbf{0}_K & \mathbf{I}_K\\
      \end{array}\right] \mathbf{C}^\top \mathbf{D}^{-1} \mathbf{y},&
    \end{aligned}
    \end{equation}
    and otherwise zero. Further,
\begin{subequations}
            
\begin{align}
    \beta_\mathbf{s} &= e^{\frac{1}{2} \left(\mathbf{y}^{\! \top} (\mathbf{N}^{-\!1} \mathbf{F}^{\! \top} (\mathbf{E}^{-\!1} \!+ \mathbf{F}^{\!\top } \mathbf{N}^{-\!1}\mathbf{F})^{\!-1} \mathbf{F} \mathbf{N}^{-\!1})\mathbf{y} - \ln \det (\mathbf{E}^{-\!1} \!+ \mathbf{F}^{\! \top} \mathbf{N}^{-\!1}\mathbf{F})  \right)} \label{eq:MMSE closed details2}& \\
    \mathbf{C} &=
  \left[\begin{array}{c  c c}
    \frac{\rho}{1 + \rho} \mathbf{\Phi}_{1,\mathcal{S}} & \frac{1}{1 + \rho} \mathbf{\Phi}_{1,\mathcal{S}} & \mathbf{0}_{M \times K} \\
    \frac{\rho}{1 + \rho} \mathbf{\Phi}_{2,\mathcal{S}} & \mathbf{0}_{M \times K} & \frac{1}{1 + \rho} \mathbf{\Phi}_{2,\mathcal{S}}  \\
  \end{array}\right],& \label{eq:accessories1} \\
  \mathbf{D} &= \left[\begin{array}{c c}
    \mathbf{\Phi}_{1,\mathcal{S}} \mathbf{\Phi}_{1,\mathcal{S}}^\top + \sigma_{w_1}^2 \mathbf{I}_M & \frac{\rho}{1 + \rho} \mathbf{\Phi}_{1,\mathcal{S}} \mathbf{\Phi}_{2,\mathcal{S}}^\top \\
    \frac{\rho}{1 + \rho} \mathbf{\Phi}_{2,\mathcal{S}} \mathbf{\Phi}_{1,\mathcal{S}}^\top &  \mathbf{\Phi}_{2,\mathcal{S}} \mathbf{\Phi}_{2,\mathcal{S}}^\top + \sigma_{w_2}^2 \mathbf{I}_M \\
  \end{array}\right],& \label{eq:accessories2} \\
  \mathbf{N} &= \left[\begin{array}{c c}
    \sigma_{w_1}^2 \mathbf{I}_M & \mathbf{0}_M \\
    \mathbf{0}_M & \sigma_{w_2}^2 \mathbf{I}_M  \\
    \end{array}\right],& \label{eq:accessories3} \\
  \mathbf{E} &= \left[\begin{array}{c c c}
    \frac{\rho}{1 + \rho} \mathbf{I}_K & \mathbf{0}_K & \mathbf{0}_K \\
    \mathbf{0}_K & \frac{1}{1 + \rho} \mathbf{I}_K & \mathbf{0}_K  \\
    \mathbf{0}_K & \mathbf{0}_K & \frac{1}{1 + \rho} \mathbf{I}_K \\
    \end{array}\right],& \label{eq:accessories4} \\
  \mathbf{F} &= \left[\begin{array}{c  c c}
     \mathbf{\Phi}_{1,\mathcal{S}} &  \mathbf{\Phi}_{1,\mathcal{S}} & \mathbf{0}_{M \times K} \\
     \mathbf{\Phi}_{2,\mathcal{S}} & \mathbf{0}_{M \times K} & \mathbf{\Phi}_{2,\mathcal{S}} \end{array}\right], \label{eq:accessories5}&
\end{align}
\end{subequations}
where $\mathbf{\Phi}_{l,\mathcal{S}} \in \mathbb{R}^{M \times K}$, $l \!\in\! \{1,2\}$, is formed by choosing the columns of $\mathbf{\Phi}_l$ indexed by the elements of support set $\mathcal{S}$.\footnote[1]{Here, for the sake of notational simplicity, we drop the dependency of the matrices $\mathbf{C}$, $\mathbf{D}$ and $\mathbf{F}$ on $\mathcal{S}$.}
\end{pro}

\begin{proof}
    The proof is given in Appendix \ref{app A}
\end{proof}

Finding an expression for the resulting MSE of the MMSE estimator \eqref{eq:MMSE closed} is analytically intractable, and there is no closed form solution. Alternatively, the resulting MSE can be lower-bounded by that of the \textit{oracle estimator} -- the ideal estimator that knows the true support set \textit{a priori}. In our studied distributed CS setup, the Bayesian oracle estimator, denoted by $\widetilde{\mathbf{X}}^{(or)}$, is derived from \eqref{eq:MMSE closed details1} given the \textit{a priori} known support, denoted by $\mathcal{S}^{(or)}$, i.e., $\widetilde{\mathbf{X}}^{(or)} = \mathbb{E}[\mathbf{X}|\mathbf{Y},\mathcal{S}^{(or)}]$. The MSE of the oracle estimator, denoted by $D_{cs}^{(or)}$, is expressed in the following proposition, which is also useful for deriving a lower-bound on end-to-end distortion shown later in \theoref{theo2}.

\begin{pro}[\textit{Oracle lower-bound}] \label{pro:oracle bound}
Let $\mathcal{S}^{(or)}$ denote the oracle-known support set for each realization of $\mathbf{X}_1$ and $\mathbf{X}_2$. Then, under \assref{ass1}, $D_{cs}$ in \eqref{eq:CS distotion} is lower-bounded as
\begin{equation} \label{eq:oracle lb}
    D_{cs} \geq D_{cs}^{(or)},
\end{equation}
where $D_{cs}^{(or)}=$
\begin{equation} \label{eq:oracle bound acc}
  1 - \frac{1}{2K} \text{Tr}\left\{
    \left[\begin{array}{c c c}
          2\mathbf{I}_K&  \mathbf{I}_K & \mathbf{I}_K \\
         \mathbf{I}_K & \mathbf{I}_K & \mathbf{0}_K \\
        \mathbf{I}_K & \mathbf{I}_K & \mathbf{I}_K \\
    \end{array}\right]\cdot \frac{1}{{N \choose K}}
     \sum_{\mathcal{S}^{(or)} \subset \mathbf{\Omega}}\mathbf{C}^\top \mathbf{D}^{\!-\!1} \mathbf{C} \right\},
\end{equation}
\end{pro}
and the matrices $\mathbf{C}$ and $\mathbf{D}$ are determined by \eqref{eq:accessories1} and \eqref{eq:accessories2}, respectively.
\begin{proof}
    The proof is given in Appendix \ref{app B}
\end{proof}

In addition to the MMSE estimator, the conditional probability density functions (pdf's) $p(\mathbf{y}_2|\mathbf{y}_1)$ and $p(\mathbf{y}_1|\mathbf{y}_2)$ also need to be considered later for optimized encoding/decdoing schemes. For the sake of completeness, we give an expression for $p(\mathbf{y}_2|\mathbf{y}_1)$ in the following proposition.
\begin{pro} [\textit{Conditional pdf}] \label{theo1-2}
    Under \assref{ass1}, the conditional pdf $p(\mathbf{y}_2|\mathbf{y}_1)$ is
    \begin{equation} \label{cond prob}
        p(\mathbf{y}_2 | \mathbf{y}_1) = \frac{\sqrt{1 + \rho}\sum_{\mathcal{S} \subset \mathbf{\Omega}} \beta_\mathcal{S}}{(\sqrt{2\pi}\sigma_{w_1})^M \sigma_{w_2}^K \sum_{\mathcal{S} \subset \mathbf{\Omega}} \gamma_\mathcal{S}},
    \end{equation}
    where $\beta_\mathcal{S}$ is specified by \eqref{eq:MMSE closed details2}, and using $\mathbf{\Psi} \triangleq [\mathbf{\Phi}_{1,\mathcal{S}} \hspace{0.15cm} \mathbf{\Phi}_{1,\mathcal{S}}]$
    \begin{equation} \label{eq:accessories cond}
    \begin{aligned}
        &\gamma_\mathcal{S} = e^{\frac{1}{2} \left(\mathbf{y}_1^\top (\frac{1}{\sigma_{w_1}^2}\mathbf{\Psi} (\mathbf{\Psi}^\top \mathbf{\Psi})^{-1} \mathbf{\Psi}^\top  - \mathbf{I}_M) \mathbf{y}_1 - \ln \det (\mathbf{\Psi}^\top \mathbf{\Psi})  \right)}. &\\
    \end{aligned}
    \end{equation}
\end{pro}

\begin{proof}
    The proof is given in Appendix \ref{app C}
\end{proof}

By symmetry, $p(\mathbf{y}_1 | \mathbf{y}_2)$ can be obtained from the same expression as in \eqref{cond prob} with the only difference that $\mathbf{y}_1$ in \eqref{eq:accessories cond} is replaced by $\mathbf{y}_2$ and $\mathbf{\Psi}$ by $[\mathbf{\Phi}_{2,\mathcal{S}} \hspace{0.15cm} \mathbf{\Phi}_{2,\mathcal{S}}]$.

Next, we show the optimization methods for encoder and decoder mappings in the system of \figref{fig:diagram_dist}. 
\subsection{Encoder Design} \label{subsec:enc design}

Let us first optimize $\textsf{E}_1$ while keeping $\textsf{E}_2$, $\textsf{D}_1$ and $\textsf{D}_2$ fixed and known. We have that
\begin{equation} \label{eq:MSE der}
\begin{aligned}
    D &= \frac{1}{2K}\sum_{i_1=0}^{\mathfrak{R}_1-1} \int_{\mathbf{y}_1 \in \mathcal{R}_{i_1}} \bigg\{ \overbrace{\mathbb{E}[\|\mathbf{X}_1 - \textsf{D}_1(J_1,J_2)\|_2^2 | \mathbf{y}_1,i_1]}^{\triangleq D_1(\mathbf{y}_1,i_1)}  &\\
    &\hspace{0.5cm} +  \underbrace{\mathbb{E}[\|\mathbf{X}_2 - \textsf{D}_2(J_1,J_2)\|_2^2 | \mathbf{y}_1,i_1]}_{\triangleq D_2(\mathbf{y}_1,i_1)}  \bigg\} p(\mathbf{y}_1) d\mathbf{y}_1,&
\end{aligned}
\end{equation}
where $p(\mathbf{y}_1)$ is the $M$-fold pdf of the measurement vector $\mathbf{Y}_1$. Since $p(\mathbf{y}_1)$ is a non-negative value, in order to optimize the mapping $\textsf{E}_1$ in the sense of minimizing $D$, it suffices to minimize the expression inside the braces in \eqref{eq:MSE der}. Thus, the optimal encoding index $i_1^\star$ is obtained by
\begin{equation} \label{eq:min index 1}
    i_1^\star = \text{arg }\underset{i_1 \in \mathcal{I}_1}{\text{min }} \Big\{ D_1(\mathbf{y}_1,i_1) + D_2(\mathbf{y}_1,i_1) \Big\}.
\end{equation}
Now, $D_1(\mathbf{y}_1,i_1)$ can be rewritten as \eqref{eq:D1}, on top of next page, where $(a)$ follows from marginalizing of the conditional expectation over $j_1$ and $j_2$ and using Markov property $J_2 \rightarrow \mathbf{Y}_1 \rightarrow I_1 \rightarrow J_1$. Also, $(b)$ follows by expanding the conditional expectation and the fact that $\mathbf{X}_1$ and $\textsf{D}_1(J_1,J_2)$ are independent conditioned on $\mathbf{y}_1,i_1,j_1,j_2$. Further, $(c)$ follows from marginalization of the expression inside the braces in $(b)$ over $i_2$ and $\mathbf{y}_2$.
\begin{figure*}[!t]
\normalsize
\setcounter{MYtempeqncnt}{\value{equation}}
\begin{equation} \label{eq:D1}
\begin{aligned}
  D_1(\mathbf{y}_1,i_1) &\triangleq \mathbb{E}[\|\mathbf{X}_1 - \textsf{D}_1(J_1,J_2)\|_2^2 | \mathbf{y}_1,i_1]&\\
    &\stackrel{(a)}{=} \sum_{j_1=0}^{\mathfrak{R}_1-1} \sum_{j_2=0}^{\mathfrak{R}_2-1} P(j_1|i_1) P(j_2|\mathbf{y}_1)\mathbb{E}[\|\mathbf{X}_1 - \textsf{D}_1(J_1,J_2)\|_2^2 | \mathbf{y}_1,i_1,j_1,j_2]& \\
    &\stackrel{(b)}{=} \mathbb{E}[\|\mathbf{X}_1\|_2^2|\mathbf{y}_1] + \sum_{j_1=0}^{\mathfrak{R}_1-1} \sum_{j_2=0}^{\mathfrak{R}_2-1} P(j_1|i_1) \Big\{P(j_2|\mathbf{y}_1) \left[\|\textsf{D}_1(j_1,j_2)\|_2^2  - 2  \mathbb{E}[\mathbf{X}_1^\top | \mathbf{y}_1,j_2] \textsf{D}_1(j_1,j_2) \right] \Big\}& \\
    &\stackrel{(c)}{=} \mathbb{E}[\|\mathbf{X}_1\|_2^2|\mathbf{y}_1] + \sum_{j_1=0}^{\mathfrak{R}_1-1} \sum_{j_2=0}^{\mathfrak{R}_2-1} P(j_1|i_1) \sum_{i_2=0}^{\mathfrak{R}_2-1} P(j_2|i_2)  \left[\|\textsf{D}_1(j_1,j_2)\|_2^2 \int_{\mathbf{y}_2 \in \mathcal{R}^{i_2}} p(\mathbf{y}_2|\mathbf{y}_1) d\mathbf{y}_2 \right.&\\
    & \left. -2 \int_{\mathbf{y}_2 \in \mathcal{R}^{i_2}} \mathbb{E}[\mathbf{X}_1^\top|\mathbf{y}_1, \mathbf{y}_2] \textsf{D}_1(j_1,j_2) p(\mathbf{y}_2|\mathbf{y}_1) d\mathbf{y}_2 \right]&
\end{aligned}
\end{equation}
\setcounter{equation}{\value{MYtempeqncnt}}
\hrulefill
\end{figure*}
\setcounter{equation}{18}

In a same fashion, $D_2(\mathbf{y}_1,i_1)$ can be parameterized similar to \eqref{eq:D1} with the only difference that $\mathbf{X}_1$ and $\textsf{D}_1(j_1,j_2)$ are replaced with $\mathbf{X}_2$ and $\textsf{D}_2(j_1,j_2)$, respectively. Following \eqref{eq:min index 1} and \eqref{eq:D1}, the MSE-minimizing encoding index, denoted by $i_1^\star$, is given by \eqref{eq:final enc1},
\begin{figure*}[!t]
\normalsize
\setcounter{MYtempeqncnt}{\value{equation}}
\begin{equation} \label{eq:final enc1}
\begin{aligned}
    i_1^\star &\!= \! \text{arg }\underset{i_1 \in \mathcal{I}_1}{\text{min }} \left\{ \sum_{j_1=0}^{\mathfrak{R}_1-1} \sum_{j_2=0}^{\mathfrak{R}_2-1} \sum_{i_2=0}^{\mathfrak{R}_2-1}  P(j_1|i_1) P(j_2|i_2) \int_{\mathcal{R}^{i_2}} \left[\|\textsf{D}(j_1,j_2)\|_2^2 \!-\! 2 \widetilde{\mathbf{x}}^\star(\mathbf{y}_1,\mathbf{y}_2)^{\!\top} \textsf{D}(j_1,j_2) \right] p(\mathbf{y}_2|\mathbf{y}_1) d\mathbf{y}_2 \right\},
\end{aligned}
\end{equation}
\setcounter{equation}{\value{MYtempeqncnt}}
\hrulefill
\end{figure*}
\setcounter{equation}{19}
where $\textsf{D}(j_1,j_2)$ and $\widetilde{\mathbf{x}}^\star(\mathbf{y}_1,\mathbf{y}_2)$ denote
\begin{equation*}
\begin{aligned}
    \textsf{D}(j_1,j_2) &\triangleq& \left[\textsf{D}_1(j_1,j_2)^\top \hspace{0.15cm} \textsf{D}_2(j_1,j_2)^\top \right]^\top  \in \mathbb{R}^{2N}, \\
    \widetilde{\mathbf{x}}^\star(\mathbf{y}_1,\mathbf{y}_2) &\triangleq& \left[\widetilde{\mathbf{x}}_1^\star(\mathbf{y}_1,\mathbf{y}_2)^\top \hspace{0.15cm} \widetilde{\mathbf{x}}_2^\star(\mathbf{y}_1,\mathbf{y}_2)^\top \right]^\top  \in \mathbb{R}^{2N}.
\end{aligned}
\end{equation*}
Note that the codevectors $\textsf{D}_1(j_1,j_2)$ and $\textsf{D}_2(j_1,j_2)$ are given, and the vectors $\widetilde{\mathbf{x}}_1^\star(\mathbf{y}_1,\mathbf{y}_2) = \mathbb{E}[\mathbf{X}_1 | \mathbf{y}_1,\mathbf{y}_2]$ and $\widetilde{\mathbf{x}}_2^\star(\mathbf{y}_1,\mathbf{y}_2) = \mathbb{E}[\mathbf{X}_2 | \mathbf{y}_1,\mathbf{y}_2]$ denote the MMSE estimators that, under \assref{ass1}, are derived in \proref{theo1}. Also, the conditional pdf $p(\mathbf{y}_2|\mathbf{y}_1)$ is given by \eqref{cond prob} in \proref{theo1-2} under \assref{ass1}. It should be mentioned that although the observation at terminal $2$, $\mathbf{y}_2$, appears in the formulation of the optimized encoder at terminal $1$, i.e, in \eqref{eq:final enc1}, it is finally integrated out.

The following remark considers the case in which sources are uncorrelated.
\begin{rem} \label{rem:dist_enc_p2p}
    When there is no correlation between sources ($\rho \rightarrow 0$), then $\mathbf{Y}_1$ and $\mathbf{Y}_2$ become independent of each other. Consequently, $J_1$ becomes independent of $J_2$, and we have the following Markov chains $\mathbf{X}_l \rightarrow \mathbf{Y}_l \rightarrow \mathbf{Y}_{l'}$ and $\mathbf{X}_l \rightarrow I_l \rightarrow I_{l'}$ ($\forall l,l' \in \{1,2\} , l \neq l'$). Then,  it is straightforward to show that the optimized encoding index \eqref{eq:final enc1} boils down to
    \begin{equation} \label{eq:final enc simple}
    i_1^\star \!\stackrel{ \rho \!\rightarrow  0}{=}\! \textrm{arg }\underset{i_1 \in \! \mathcal{I}_1}{\textrm{min}} \left\{ \! \sum_{j_1=0}^{\mathfrak{R}_1-1} \!  \! \! P(j_1|i_1) \left( \left\| \textsf{D}_{1}(j_1) \right\|_2^2 \!-\! 2 \widetilde{\mathbf{x}}_1^\star(\mathbf{y}_1)^{\!\top} \textsf{D}_1(j_1)\right) \! \right\}
\end{equation}
    which is the optimized encoding index for the point-to-point vector quantization of CS measurements over a noisy channel, cf. \cite[eq. (7)]{13:Shirazinia_ICASSP}.
\end{rem}

\subsection{Decoder Design} \label{subsec:dec design}
Assuming all encoders and decoder $l'$ $(l' \neq l)$ are fixed, the MSE-minimizing decoder is given by
\begin{equation} \label{eq:final dec}
    \textsf{D}_l^\star(j_1,j_2) = \mathbb{E}[\mathbf{X}_l | j_1,j_2],  \hspace{0.1cm} j_l \in \mathcal{I}_l, l \in \{1,2\}.
\end{equation}
Using the Bayes' rule, it follows that
\begin{equation} \label{eq:bayes}
P(i_1,i_2|j_1,j_2) = \frac{P(j_1|i_1)P(j_2|i_2)P(i_1,i_2)}{\sum_{i_1} \sum_{i_2} P(j_1|i_1)P(j_2|i_2)P(i_1,i_2)},
\end{equation}
where $P(i_1,i_2) = \text{Pr}(\mathbf{Y}_1 \in \mathcal{R}^{i_1},\mathbf{Y}_2 \in \mathcal{R}^{i_2})$. Now, marginalizing the conditional expectation \eqref{eq:final dec} over $i_1$ and $i_2$ and applying \eqref{eq:bayes}, we obtain $\textsf{D}_l^\star(j_1,j_2)=$
\begin{equation} \label{eq:param dec1}
\begin{aligned}
      \frac{\sum_{i_1,i_2} P(j_1|i_1) P(j_2|i_2) \int_{\mathcal{R}^{i_1}} \int_{\mathcal{R}^{i_2}} \widetilde{\mathbf{x}}_l^\star (\mathbf{y}_1,\mathbf{y}_2)  p(\mathbf{y}_1,\mathbf{y}_2) d\mathbf{y}_1 d \mathbf{y}_2}{\sum_{i_1,i_2} P(j_1|i_1) P(j_2|i_2) \int_{\mathcal{R}^{i_1}} \int_{\mathcal{R}^{i_2}} p(\mathbf{y}_1,\mathbf{y}_2) d\mathbf{y}_1 d \mathbf{y}_2}.
\end{aligned}
\end{equation}

We note the following remark regarding the case where the sources are uncorrelated.
\begin{rem}
    In a scenario where the sources are uncorrelated ($\rho \rightarrow 0$), due to the same reasoning stated in \remref{rem:dist_enc_p2p}, the optimized codevectors in the studied distributed scenario, i.e., \eqref{eq:param dec1}, boil down to
    \begin{equation} \label{eq:param dec1 simple}
        \textsf{D}_l^\star(j_l) \stackrel{ \rho \!\rightarrow  0}{=} \frac{\sum_{i_l} P(j_l|i_l) \int_{\mathcal{R}^{i_l}} \widetilde{\mathbf{x}}_l^\star (\mathbf{y}_l)  p(\mathbf{y}_l) d\mathbf{y}_l}{\sum_{i_l} P(j_l|i_l) \int_{\mathcal{R}^{i_l}}  p(\mathbf{y}_l)  d \mathbf{y}_l}.
    \end{equation}
    which is the optimized codevectors for the point-to-point vector quantization of CS measurements over a noisy channel, cf. \cite[eq. (11)]{13:Shirazinia_ICASSP}.
\end{rem}

\subsection{Training Algorithm} \label{sec:training}
In this section, we develop a practical VQ encoder-decoder training algorithm for the studied distributed system.

The necessary optimal conditions for the encoder in \eqref{eq:final enc1} (and its equivalence $i_2^\star$) and the decoder in \eqref{eq:param dec1} can be combined in an \textit{alternate-iterate} procedure in order to design distributed VQ encoder-decoder pairs for CS. We choose the order to optimize the mappings as: \textit{1) the first encoder, 2) the first decoder, 3) the second encoder and 4) the second decoder} as shown in \algref{alg:Lloyd_Dist}.

\begin{algorithm}
\caption{: Training algorithm for distributed vector quantization of CS measurements over noisy channels.}\label{alg:Lloyd_Dist}
\begin{algorithmic}[1]
\STATE{\textbf{input:} measurement vector $\mathbf{y}_l$, channel probabilities: $P(j_l|i_l)$, quantization rate: $R_l$, $l \in \{1,2\}$}
\STATE{\textbf{initialize: } $\textsf{D}_l$ , $l \in \{1,2\}$}
\REPEAT
    \STATE{Fix the second encoder and the decoders, and find the optimal index for the first encoder using \eqref{eq:final enc1}.}
    \STATE{Fix the encoders and the second decoder, and find the optimal codevectors for the first decoder using \eqref{eq:param dec1}.}
    \STATE{Fix the first encoder and the decoders, and find the optimal index for the second encoder using equivalence of \eqref{eq:final enc1}.}
    \STATE{Fix the encoders and the first decoder, and find the optimal codevectors for the second decoder using \eqref{eq:param dec1}.}
\UNTIL{convergence}
\STATE{\textbf{output: } $\textsf{D}_l$ and $\mathcal{R}_{i_l}$, $l \in \{1,2\}$ }
\end{algorithmic}
\end{algorithm}

The following remarks can be taken into consideration for implementation of \algref{alg:Lloyd_Dist}.
\begin{itemize}
    \item In order to initialize \algref{alg:Lloyd_Dist} in step (2), codevectors for the first and the second decoders might be chosen as sparse random vectors (with known statistics) to mimic the behavior of the sources. Furthermore, the convergence of the algorithm in step (8) may be checked by tracking the MSE, and iterations are terminated when the relative improvement is small enough. By construction and ignoring issues such as numerical precision, the iterative design always converges to a local optimum since when the necessary optimal criteria in steps (4)-(7) of \algref{alg:Lloyd_Dist} are invoked, the performance can only leave unchanged or improved, given the updated indexes and codevectors. This is a common rationale behind the proof of convergence for such iterative algorithms (see e.g. \cite[Lemma 11.3.1]{91:Gersho}).

       \item In step (4) of \algref{alg:Lloyd_Dist}, we need to compute the integral in \eqref{eq:final enc1} which consists of the MMSE estimator \eqref{eq:MMSE closed} and the conditional probability \eqref{cond prob}. The expressions for these parameters are derived analytically in \proref{theo1} and \proref{theo1-2} under \assref{ass1}. However, the integral in \eqref{eq:final enc1} cannot be solved in closed form, and requires approximation.
           Let us focus on evaluating the integral in \eqref{eq:final enc1}. We rewrite the integral as
            \begin{equation} \label{eq:modify enc1}
            \begin{aligned}
                &\int_{\mathcal{R}^{i_2}} \! \left[\|\textsf{D}(j_1,j_2)\|_2^2 \!-\! 2 \widetilde{\mathbf{x}}^\star(\mathbf{y}_1,\mathbf{y}_2)^\top \textsf{D}(j_1,j_2) \right] p(\mathbf{y}_2|\mathbf{y}_1) d\mathbf{y}_2& \\
                &= \! \|\textsf{D}(j_1,j_2)\|_2^2 P(i_2 | {\mathbf{y}}_1) \!-\! 2 P(i_2 | {\mathbf{y}}_1) \mathbb{E} [\mathbf{X}^\top | {\mathbf{y}}_1 , i_2] \textsf{D}(j_1,j_2)& \\
                &\approx \! \|\textsf{D}(j_1,j_2)\|_2^2 P(i_2 | \check{\mathbf{y}}_1) \!-\! 2 P(i_2 | \check{\mathbf{y}}_1) \mathbb{E} [\mathbf{X}^\top | \check{\mathbf{y}}_1 , i_2] \textsf{D}(j_1,j_2)&
            \end{aligned}
            \end{equation}
            where we have approximated $\mathbf{y}_1$ in $M-$dimensional continuous space with a $M-$dimensional vector $\check{\mathbf{y}}_1$ belongs to a discrete space. This is performed by scalar-quantizing each entry of $\mathbf{y}_1$ using $r_y$-bit nearest-neighbor coding. Here, $r_y$ denotes the number of quantization bits per measurement entry, and determines the resolution of the measurements. For simplicity of implementation, we use the codeponits optimized for a Gaussian RV (with zero mean and variance $K/M$) for each measurement entry using the LBG algorithm \cite{80:LBG}. Hence, $\mathbf{y}_1$ is discretized using this pre-quantization method. Also, $P(i_2 | \check{\mathbf{y}}_1) \triangleq \text{Pr}\{I_2 = i_2 | \check{\mathbf{Y}}_1 = \check{\mathbf{y}}_1\}$ indicates a transition probability that can be calculated by counting the number of transitions from $\check{\mathbf{y}}_1$ to $i_2$ over total occurrences of $\check{\mathbf{y}}_1$'s. Note that this probability can be computed off-line and be available at the first encoder. In order to evaluate the conditional mean $\mathbb{E}[\mathbf{X}^\top | \check{\mathbf{y}}_1 , i_2]$ in \eqref{eq:modify enc1}, we generate samples of $\mathbf{X}_1$ and $\mathbf{X}_2$, and then take average over those samples that have resulted in the quantized value $\check{\mathbf{y}}_1$ and the quantization index $i_2$. Using this trick, the conditional mean is replaced by a look-up table that can be calculated off-line and stored. Here, we emphasize that the value of $i_2$ in online phase of quantization is not required at the first terminal since it is summed out because of the summation over $i_2$ in \eqref{eq:final enc1}. For all practical purpose, the approximation in \eqref{eq:modify enc1} is used instead of the integral in \eqref{eq:final enc1}. Also, note that we can use the same modification, discussed above, in step (6) of \algref{alg:Lloyd_Dist}.

            \item Using the discussed modifications, for a encoding given index $i_l$ and the pre-quantized value $\check{\mathbf{y}}_l$ ($l \in \{1,2\}$), the encoder computational complexity grows at most like $\mathcal{O}(2^{R_1 + R_2})$. We stress that, in this paper, we used VQ at each terminal since it is theoretically the optimal coding strategy for a block (vector). Therefore, we have not sacrificed performance to reduce complexity, which is not the scope of the current work. However, using structured quantizers, such as tree-structured VQ and multi-stage VQ \cite{91:Gersho,93:Phamdo}, the encoding complexity of VQ can be reduced, but this is achieved at the expense of further performance degradation.

            \item In steps (5) and (7) of \algref{alg:Lloyd_Dist}, we need to compute the codevectors \eqref{eq:param dec1}. Note that although $\mathbb{E}[\mathbf{X}_l | j_1,j_2]$, $l \in \{1,2\}$, can be calculated analytically from \eqref{eq:param dec1}, it requires massive integrations of non-linear functions. Therefore, we calculate $\mathbb{E}[\mathbf{X}_l | j_1,j_2]$ empirically by generating Monte-Carlo samples of $\mathbf{X}_l$, and then take average over those samples which have led to the noisy quantized indexes $j_1$ and $j_2$.
\end{itemize}

In the next section, we offer insights into the performance characteristics of the distributed system shown in \figref{fig:diagram_dist}.

\section{Analysis of MSE} \label{sec:analysis}
We can rewrite the end-to-end MSE, in \eqref{eq:MSE}, as
\begin{equation} \label{eq:MSE decomp}
\begin{aligned}
    D &\stackrel{(a)}{=} \! \frac{1}{2K}\sum_{l=1}^2 \mathbb{E}[\|\mathbf{X}_l - \widetilde{\mathbf{X}}_l^\star\|_2^2] + \frac{1}{2K}\sum_{l=1}^2 \mathbb{E}[\|\widetilde{\mathbf{X}}_l^\star - \widehat{\mathbf{X}}_l\|_2^2]&\\
    &\triangleq D_{cs} + D_q,&
\end{aligned}
\end{equation}
where $\widetilde{\mathbf{X}}_l^\star \triangleq \mathbb{E}[\mathbf{X}_l|\mathbf{Y}]$ denotes a RV representing the MMSE estimator, and $D_{cs}$ and $D_q$, respectively, denote the CS distortion (MSE) and quantized transmission distortion (MSE). In \eqref{eq:MSE decomp}, $(a)$  holds due to orthogonality of CS reconstruction error (i.e., $\mathbf{X}_l - \widetilde{\mathbf{X}}_l^\star$) and quantized transmission error (i.e., $\widetilde{\mathbf{X}}_l^\star - \widehat{\mathbf{X}}_l$). This can be shown based on the definition of the MMSE estimator $\widetilde{\mathbf{X}}_l^\star$ and the Markov property $\mathbf{X}_l \rightarrow (J_1,J_2) \rightarrow \widehat{\mathbf{X}}_l$, $l \in \{1,2\}$. Next, we use \eqref{eq:MSE decomp} in order to develop a lower-bound on $D$.

\begin{theo} [\textit{Lower-bound on end-to-end MSE}] \label{theo2}
Consider the two-terminal distributed system in \figref{fig:diagram_dist} under \assref{ass1}. Let the total quantization rate be $R = R_1 + R_2$ bits/vector, and the correlation ratio between sources be $\rho$. Then the asymptotic (in quantization rate) end-to-end MSE \eqref{eq:MSE} is lower-bounded as
\begin{equation} \label{eq:total lb}
\begin{aligned}
	D >  \max \left\{D_q^{(or)} , D_{cs}^{(or)} \right\},
\end{aligned}
\end{equation}
where $D_{q}^{(or)} =$
\begin{equation} \label{eq:lb sum}
\begin{aligned}
	&   \sqrt{\left(1- \frac{\rho^2}{(1+\rho)^2}\right) 2^{\frac{-2\left(R -  \log_2 {N \choose K}\right)}{K}} + \frac{\rho^2}{(1+\rho)^2} 2^{\frac{-4\left(R -  \log_2 {N \choose K}\right)}{K}}},&
\end{aligned}
\end{equation}
and $D_{cs}^{(or)} =$
\begin{equation} \label{eq:oracle bound acc}
      1 - \frac{1}{2K} \text{Tr}\left\{
    \left[\begin{array}{c c c}
          2\mathbf{I}_K&  \mathbf{I}_K & \mathbf{I}_K \\
         \mathbf{I}_K & \mathbf{I}_K & \mathbf{0}_K \\
        \mathbf{I}_K & \mathbf{I}_K & \mathbf{I}_K \\
    \end{array}\right]\cdot \frac{1}{{N \choose K}}
     \sum_{\mathcal{S}^{(or)} \subset \mathbf{\Omega}}\mathbf{C}^\top \mathbf{D}^{-1} \mathbf{C} \right\},
\end{equation}
where $\mathcal{S}^{(or)}$ is an oracle support in the set $\mathbf{\Omega}$, and the matrices $\mathbf{C}$ and $\mathbf{D}$ are specified by \eqref{eq:accessories1} and \eqref{eq:accessories2}, respectively.
\end{theo}
\begin{proof}
    The proof is given in Appendix \ref{app D}
\end{proof}

The following remarks can be made with reference to the lower-bound \eqref{eq:total lb} in \theoref{theo2}.
\begin{itemize}
    \item The term $D_{cs}^{(or)}$ in \eqref{eq:total lb} is the contribution of the CS distortion of the MMSE estimator $\widetilde{\mathbf{x}}_l^\star(\mathbf{y})$ $(l \in \{1,2\})$ derived in \eqref{eq:MMSE closed}. Further, the term $D_{q}^{(or)}$ reflects the contribution of quantized transmission distortion. When the CS measurements are noisy, it can be verified that as the sum rate $R = R_1 + R_2$ increases, the lower-bound in \eqref{eq:total lb} saturates since $D_q^{(or)}$ decays exponentially, however, $D_{cs}^{(or)}$ becomes constant by quantization rate (see  \figref{fig:MSE_vs_R} later in the numerical experiments). Hence, the end-to-end MSE, $D$, can, at most, approaches an MSE floor equivalent to $D_{cs}^{(or)}$.
     \item When CS measurements are clean, and number of measurements are sufficient such that $D_{cs} = 0$ in \eqref{eq:MSE decomp} (that is $\widetilde{\mathbf{X}}_l^{\star} = \mathbf{X}_l$, $l \in \{1,2\}$), then it can be shown that $D \geq D_q^{(or)}$ (see \eqref{eq:quant dist total} in the proof of \theoref{theo2}). In this case, the end-to-end MSE can asymptotically  decay at most $-6/K$ dB per total bit (i.e., $R$) corresponding to the case where $\rho \rightarrow \infty$. However, if $\rho \rightarrow 0$, the end-to-end MSE cannot decay steeper than $-3/K$ dB per total bit.
     \item The source correlation, in terms of $\rho$, also plays an important role on the level of the lower-bound in \eqref{eq:total lb}. It can be shown that increasing $\rho$ improves CS reconstruction performance $D_{cs}^{(or)}$ (see \figref{fig:Bayes_MMSE} later in the numerical experiments). Further, by taking the first derivative of $D_q^{(or)}$ in \eqref{eq:lb sum} with respect to $\rho$, it can be verified that the derivative is always negative. This means that as the source correlation $\rho$ increases, the lower-bound would decrease. Therefore, the correlation between sources can be useful in order to reduce the lower-bound and total distortion. This behavior can be also seen from simulation results in the next section.
\end{itemize}

\section{Numerical Experiments} \label{sec:numerical}
In this section, we first give experimental steps, and then show the simulation and analytical results.

\subsection{Experimental Setups} \label{sec:setup}

The sources $\mathbf{X}_1$ and $\mathbf{X}_2$ are generated randomly according to \assref{ass1}. The correlation ratio is also adjusted by \eqref{eq:corr ratio}.

For the purpose of reproducible research, and due to the reason that structured deterministic sensing matrices are practically implementable, due to hardware considerations, rather than random sensing matrices, we choose a deterministic construction for the sensing matrices \cite{11:Duarte}. More specifically, the sensing matrices $\mathbf{\Phi}_1$ and $\mathbf{\Phi}_2$ are produced by choosing the first (indexed from the first row downwards) and the last (indexed from the last row upwards) $M$ rows of a $N \times N$ discrete cosine transform (DCT) matrix. Then, the columns of the resulting matrices are normalized to unit-norm. Note that once the sensing matrix is generated, it remains fixed. Although the sensing matrices are deterministic, we believe that simulation trends are the same for random matrix generation.

In order to measure the level of under-sampling, we define the \textit{measurement rate} $0 < \alpha < 1$ as
\begin{equation*}
    \alpha \triangleq M/N.
\end{equation*}
Assuming Gaussian measurement noise vector, we define the signal-to-measurement noise ratio (SMNR) at terminal $l \in \{1,2\}$ as
\begin{equation*} \label{eq:SMNR}
    \text{SMNR}_l \triangleq \frac{\mathbb{E}[\|\mathbf{X}_l\|_2^2]}{\mathbb{E}[\|\mathbf{W}_l\|_2^2]} = \frac{K}{M\sigma_{w_l}^2}.
\end{equation*}

For the simulation results associated with noisy channels, we implement a binary symmetric channel (BSC) with bit cross-over probability $0 \leq \epsilon \leq 0.5$ specified by transition probability
\begin{equation} \label{eq:BSC}
     P(j | i) = \epsilon^{H_\mathfrak{R}(i,j)} (1 - \epsilon)^{\mathfrak{R} - H_\mathfrak{R}(i,j)},
\end{equation}
where $\epsilon$ represents bit cross-over probability (assumed known), and $H_\mathfrak{R}(i,j)$ denotes the Hamming distance between $\mathfrak{R}$-bit binary codewords representing the channel input and output indexes $i$ and $j$.

\subsection{Experimental Results} \label{sec:results}
In the following, through numerical experiments, we first offer insights regarding the impact of correlation between sources $\rho$ and measurement rate $\alpha$ on the CS distortion. Then, using the distributed design, we evaluate the end-to-end performance in terms of correlation ratio $\rho$, compression resources (measurement rate $\alpha$ and quantization rate $R$), and channel bit cross-over probability $\epsilon$.

\subsubsection{CS Distortion}
The performance is tested using the CS distortion criterion, $D_{cs}$ in \eqref{eq:CS distotion}. We set the source dimension to $N=10$, sparsity level $K=2$ and $\text{SMNR}_l = 10$ dB ($l \in \{1,2\}$). We randomly generate $2 \times 10^4$ samples of each sparse source vector, and empirically compute $D_{cs}$ in \eqref{eq:CS distotion} using the MMSE estimator \eqref{eq:MMSE closed}. The results are illustrated in \figref{fig:Bayes_MMSE} as a function of the correlation ratio $\rho$ for different values of measurement rate, i.e., $\alpha=\frac{3}{10}, \frac{4}{10}, \frac{5}{10}, \frac{6}{10}$. The analytical lower-bound in \eqref{eq:oracle bound acc} corresponding to the measurement rate $\alpha = \frac{6}{10}$ is also demonstrated. From \figref{fig:Bayes_MMSE}, we observe that increasing number of CS measurements improve the performance which is expected since the sources are estimated from more observations. Another interesting point is that $D_{cs}$ varies significantly by changing the correlation ratio $\rho$; for example, there is $4.5$ dB performance improvement (corresponding to the curve $\alpha=\frac{6}{10}$) from the case where the sources are almost uncorrelated ($\rho = 10^{-3}$) to the one that they are highly correlated ($\rho = 10^{3}$). This behavior is reflected from the oracle lower-bound as well. This is due to the fact that at low correlation, the measurement vectors become uncorrelated, therefore there is no gain obtained by, e.g., estimation of $\mathbf{X}_1$ from observations at the second terminal, i.e., $\mathbf{y}_2$. On the other hand, when the sources are highly correlated, the estimation procedure tends to estimating a single source $\mathbf{\Theta}$ from $2M$ number of observations, i.e., $\mathbf{y}_1$ and $\mathbf{y}_2$. Finally, it should be noted that the gap between the curve corresponding to $\alpha = \frac{6}{10}$ and the
oracle lower-bound is due to imperfect knowledge of exact support set.
\begin{figure}
    \begin{center}
     \includegraphics[width=\columnwidth,height=8cm]{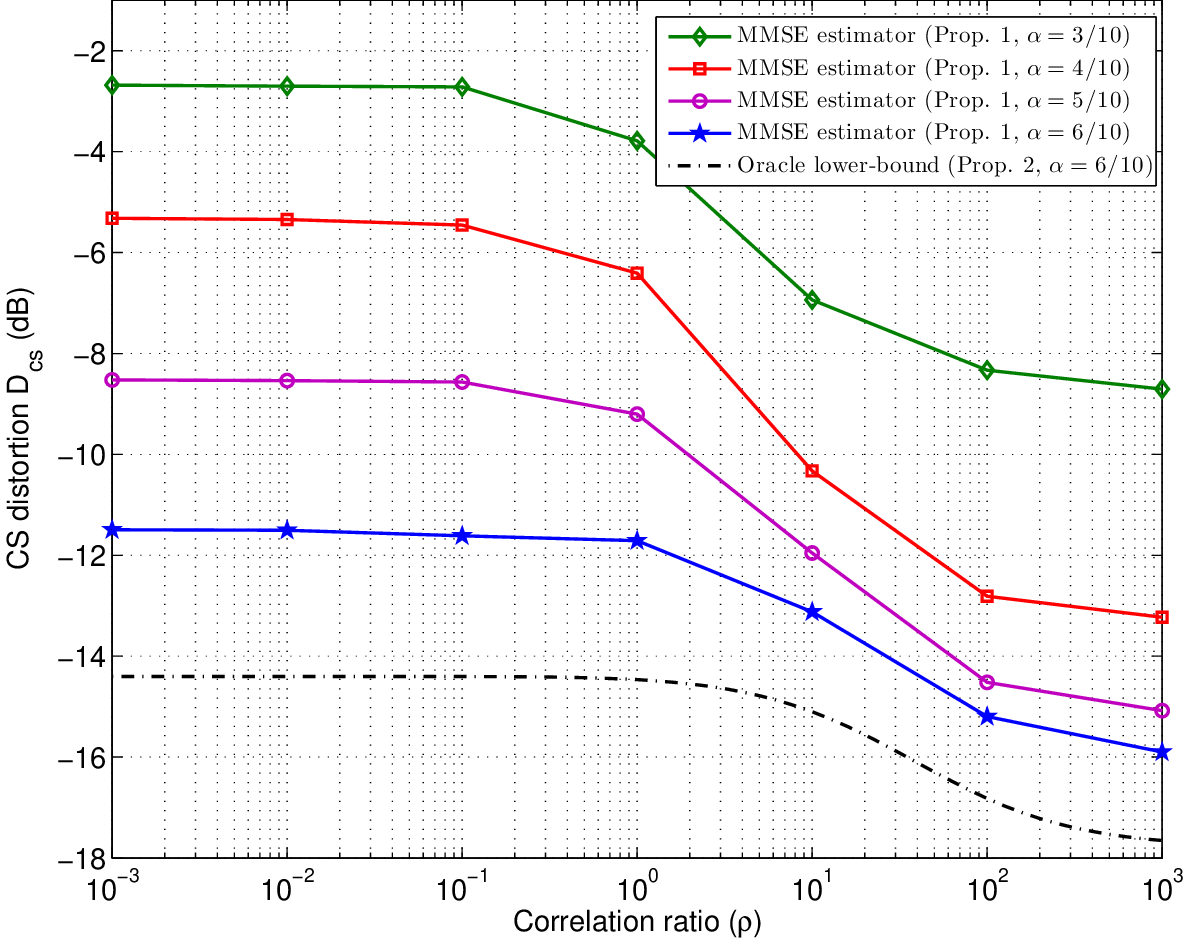}\\
     \caption{CS distortion $D_{cs}$ (in dB) vs. correlation ratio $\rho$. The parameters are chosen as $N=10$, $K=2$, and $\text{SMNR}_1 = \text{SMNR}_2 = 10 \text{dB}$.}
     \label{fig:Bayes_MMSE}
    \end{center}
\end{figure}

\subsubsection{End-to-end Distortion}

The performance is now tested using the end-to-end MSE, $D$ in \eqref{eq:MSE}. It is well-known that VQ is theoretically the optimal block coding strategy, but it suffers in exponential complexity with the dimension of source and bit rate per sample. Therefore, in our simulations, we are compelled to use a low-dimensional setup.\footnote[1]{At this point, readers are reminded that small dimensions and rates are required in order to enable the computations of the full-search VQ. For higher dimension and rate, the usual approach is to design sub-optimal structured VQ, such as multistage VQ, tree-structured VQ, etc. This paper is our first coordinated effort to bring channel-robust VQ and CS together in a distributed setup. To deal with complexity, the design of structured VQ in this setup remains open for further research.} All simulations, both in training and in performance evaluation, are performed by using $3 \times 10^5$ realizations of the source vectors. Further, the vectors $\mathbf{y}_1$ and $\mathbf{y}_2$ are pre-quantized (as discussed in \secref{sec:training}) using $r_y=3$ bits per measurement entry in order to obtain $\check{\mathbf{y}}_1$ and $\check{\mathbf{y}}_2$, respectively.

In our first experiment, we demonstrate the effect of source correlation $\rho$ and measurement rate $\alpha$ on the performance. We use the simulation parameter set ($N = 10, K = 2, R = R_1 + R_2 = 10$ bits/vector with $R_1 = R_2$), and assume noiseless communication channels and clean measurements. We vary the correlation ratio from very low $(\rho = 10^{-3})$ to very high values $(\rho = 10^3)$, and compare the simulation results with the lower-bound (corresponding to the curve $\alpha = \frac{6}{10}$) derived in \eqref{eq:total lb} of \theoref{theo2}. The results are shown in \figref{fig:MSE_rho_10} for various values of measurement rate, i.e., $\alpha = \frac{3}{10}, \frac{4}{10}, \frac{5}{10}, \frac{6}{10}$. From \figref{fig:MSE_rho_10}, we observe that the higher the correlation (i.e., larger $\rho$) is, the better the performance gets. This behavior was previously observed from the curves in \figref{fig:Bayes_MMSE} with the objective of CS distortion. As would be expected, at a fixed quantization rate $R$ and correlation ratio $\rho$, increasing $\alpha$ improves the performance, and the curves approach the lower-bound. In this simulation setup, $D_{cs}^{(or)} \ll D_{q}^{(or)}$ (see \eqref{eq:total lb}), hence, the lower-bound mainly shows the contribution of quantization distortion.

\begin{figure}
  \begin{center}
  \includegraphics[width=\columnwidth,height=8cm]{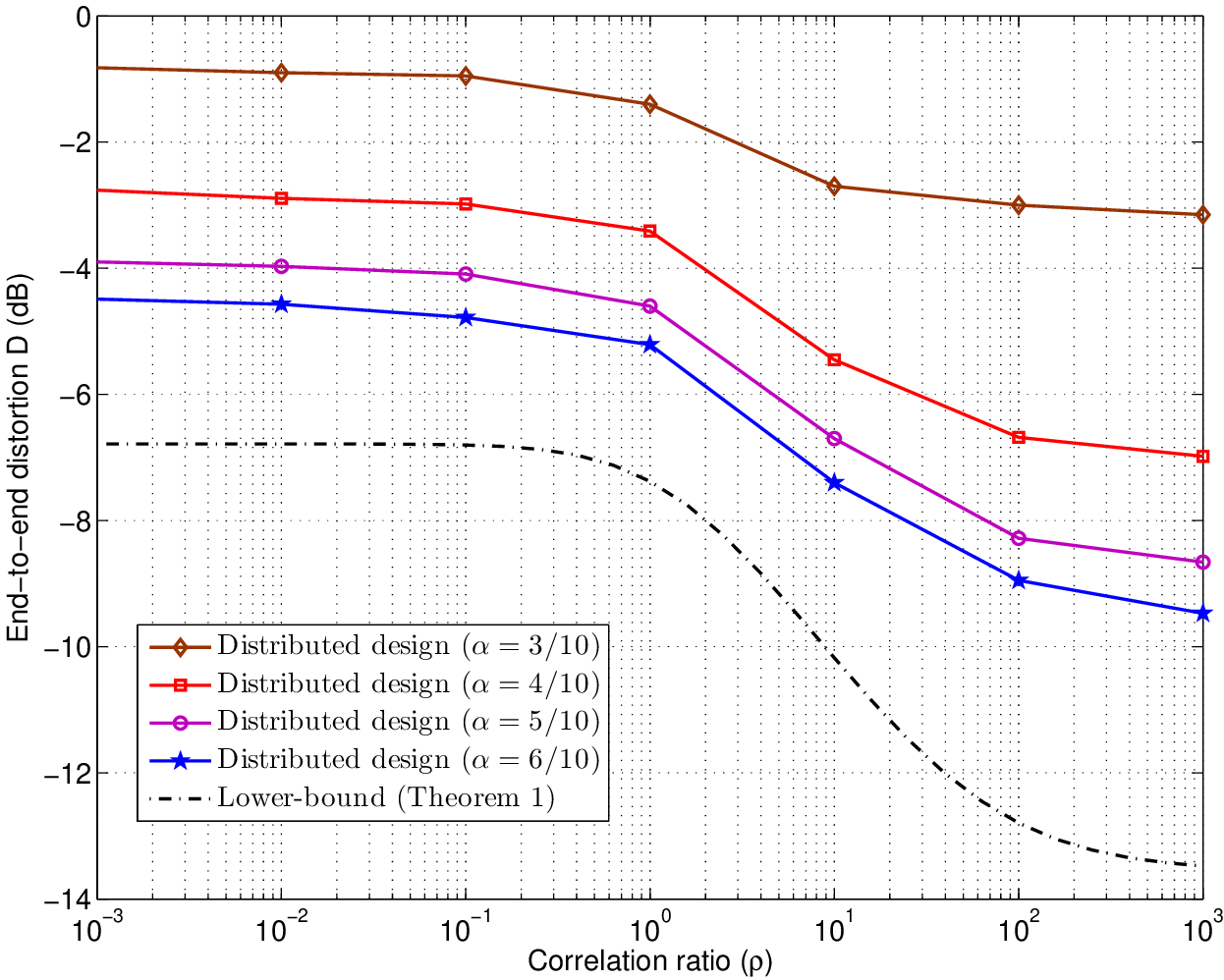}\\
  \caption{End-to-end distortion ($D$ in dB) vs. correlation ratio $\rho$ using the proposed design scheme along with the lower-bound \eqref{eq:total lb} of \theoref{theo2} for different values of measurement rate $\alpha$. The parameters are chosen as $N=10$, $K=2$ and $R = R_1 + R_2 = 10$ bits/vector for clean measurements and noiseless channels.}
  \label{fig:MSE_rho_10}
  \end{center}
\end{figure}

Next, we investigate how the performance varies by quantization rate. We use the simulation parameter set $(N = 10, K = 2, M = 5, \text{ SMNR}_1 = \text{SMNR}_2 = 10 \text{ dB})$, and assume noiseless communication channels. In \figref{fig:MSE_vs_R}, we illustrate the end-to-end MSE of the proposed distributed design method as a function of total quantization rate $R = R_1 + R_2$ (with $R_1 = R_2$) for different values of correlation ratios: $\rho = 1$ (low-correlated sources), $\rho = 10$ (moderately-correlated sources) and $\rho = 10^3$ (highly-correlated sources). The simulation curves are compared with the lower-bound in \eqref{eq:total lb} corresponding to $\rho=1, 10, 10^3$. From \figref{fig:MSE_vs_R}, we observe that the performance improves by increasing quantization rate. Moreover, increasing correlation between sources reduces the MSE as observed from the previous experiments too. The gap between the simulation curves and their respective lower-bounds is due to the reason that when CS measurements are noisy, the MMSE estimators $\widetilde{\mathbf{X}}_l^\star$ $(l \in \{1,2\})$ become far from Gaussian vectors within support. Hence, the simulation curves do not decay as steep as their corresponding lower-bounds which are derived under the optimistic assumption that the source $\mathbf{X}_l$ is available for coding. It should be also noted that as quantization rate increases, all the simulation curves will eventually approach to their respective MSE floors, specified by $D_{cs}$. This is reflected from the lower-bounds in \figref{fig:MSE_vs_R}, where each attains an MSE floor equivalent to $D_{cs}^{(or)} \leq D_{cs}$.

\begin{figure}
  \begin{center}
  \includegraphics[width=\columnwidth,height=8cm]{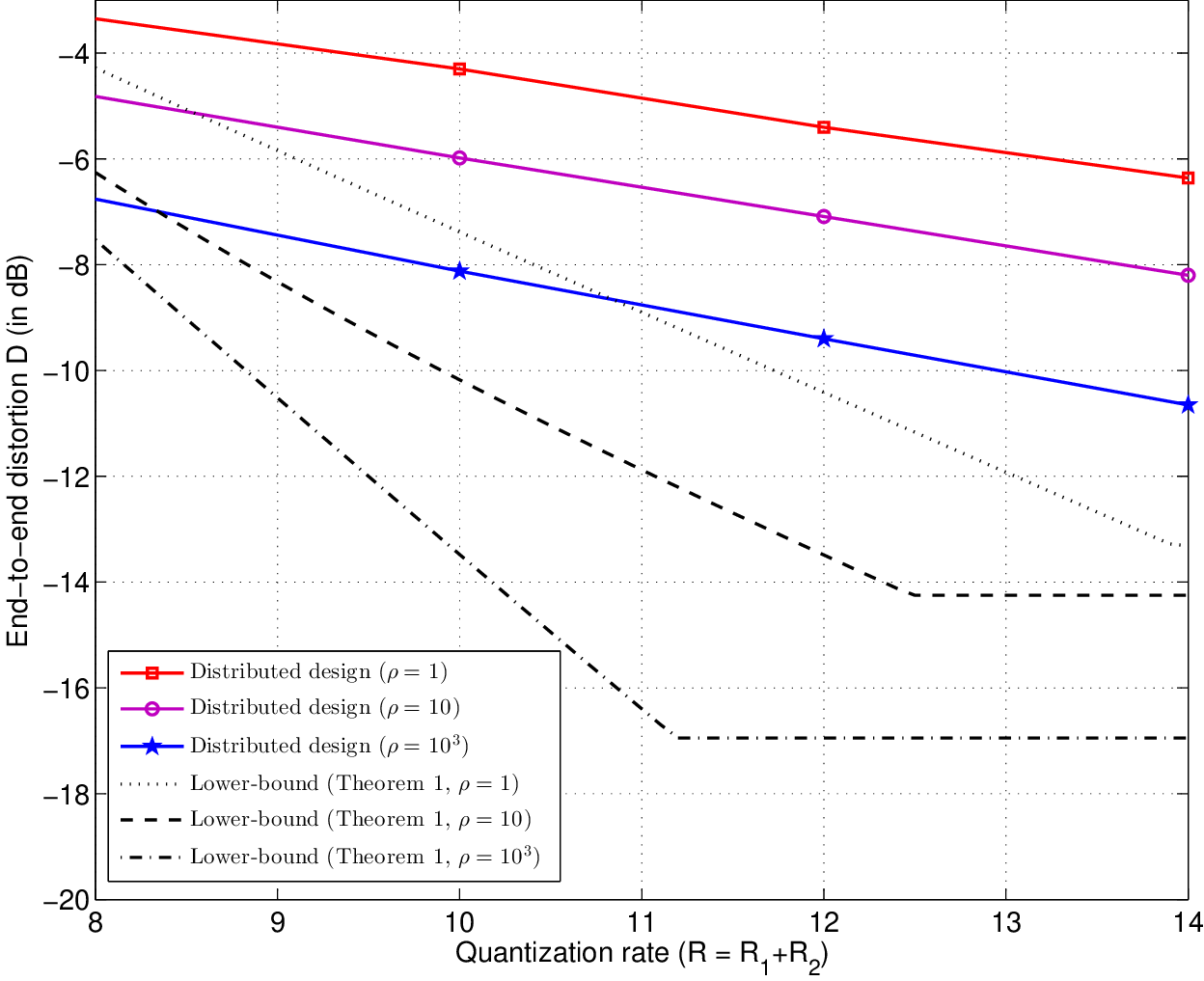}\\
  \caption{End-to-end distortion ($D$ in dB) vs. quantization rate $R = R_1 + R_2$ (in bits/vector) using the proposed design scheme along with the lower-bound \eqref{eq:total lb} for different values of correlation ratio $\rho$. The parameters are chosen as $N=10$, $K=2$ and $M = 5$ for noisy CS measurements, with $\text{SMNR}_1 = \text{SMNR}_2 = 10$ dB, and noiseless channels.}
  \label{fig:MSE_vs_R}
  \end{center}
\end{figure}

In our final experiment, we study the impact of channel noise on the performance. In \figref{fig:MSE_eps}, we assess the MSE of the distributed design as a function of channel bit cross-over probability $\epsilon$ (which is the same for channels at both terminals) using the simulation setup $(N=10, K=2, M=5, R=R_1+ R_2=10 \text{ bits/vector, with } R_1 = R_2)$ and for two values of correlation ratio, $\rho = 1,10^3$. Further, measurement noise is negligible. In order to demonstrate the efficiency of the distributed design scheme, in \figref{fig:MSE_eps}, we also plot the performance of a centralized design of VQ for CS measurements presented in \cite{13:Shirazinia_ICASSP}. A \textit{centralized} scheme provides benchmark performance where the concatenated measurement vector $\mathbf{Y} = [\mathbf{Y}_1^\top \hspace{0.1cm} \mathbf{Y}_1^\top]^\top \in \mathbb{R}^{10}$ is encoded using a VQ encoder with $R = 10$ bits/vector, and the concatenated source $\mathbf{X} = [\mathbf{X}_1^\top \hspace{0.1cm} \mathbf{X}_2^\top]^\top \in \mathbb{R}^{20}$ (with $4$ non-zero coefficients) is reconstructed at the decoder. From the simulation curves in \figref{fig:MSE_eps}, it can be observed that degrading channel condition increases MSE. However, the channel-robust VQ design provides robustness against channel noise by considering channel through its design. At high channel noise, the centralized design provides a slightly more robust performance compared to the distributed design, particularly, for the curve corresponding to low correlation. A potential reason is that the centralized design operates on joint source-channel codes of length $10$ bits, while the distributed design has the encoded index of length $5$ bits at each terminal. However, at high correlation ratio, it can be seen that the performance of the distributed design closely follows that of the centralized approach. By comparing the performance of the distributed design at $\rho=1$ and $\rho=10^3$, it is revealed that correlation between sources is also useful in providing a better performance in noisy channel scenarios. At very high channel noise level, the performance of the designs for $\rho=1$ and $\rho = 10^3$ approaches together. To interpret this behavior, let us consider an extreme case where $\epsilon \rightarrow 0.5$. For a BSC with transition probabilities \eqref{eq:BSC}, this gives $P(j_l | i_l) \rightarrow \left(1/2\right)^{\mathfrak{R}_l} = 2^{-2^{R_l}}, \hspace{0.25cm} \forall i_l,j_l, l \in \{1,2\}$, and according to \eqref{eq:bayes}, this implies that $P(i_1,i_2|j_1,j_2) \rightarrow P(i_1,i_2)$, $\forall j_1,j_2$. Studying \eqref{eq:final dec}, we get $\textsf{D}_l^\star (j_1,j_2) \rightarrow \mathbb{E}[\mathbf{X}_l]$, $\forall j_l, l\in \{1,2\}$. This means that all the codevectors become equal. Further, studying the expression \eqref{eq:final enc1} for the optimized encoding index, we obtain
\begin{equation*}
    i_l^\star \! \rightarrow \text{arg }\underset{i_l \in \mathcal{I}_l}{\text{min }} \left\{\|\mathbb{E}[\mathbf{X}_l]\|_2^2 - 2\mathbb{E}[\mathbf{X}_l^\top | \mathbf{y}_l] \mathbb{E}[\mathbf{X}_l] \right\}, \forall j_l, l \in \{1,2\}.
\end{equation*}
This implies that we have only one non-empty encoding region. Hence, at very high channel noise, only one index is transmitted -- irrespective of the input $\mathbf{X}_l$ and correlation between sources -- and the decoder produces the expected value of the source $\mathbb{E}[\mathbf{X}_l]$ for all received indexes from the channel.

\begin{figure}
  \begin{center}
  \includegraphics[width=\columnwidth,height=8cm]{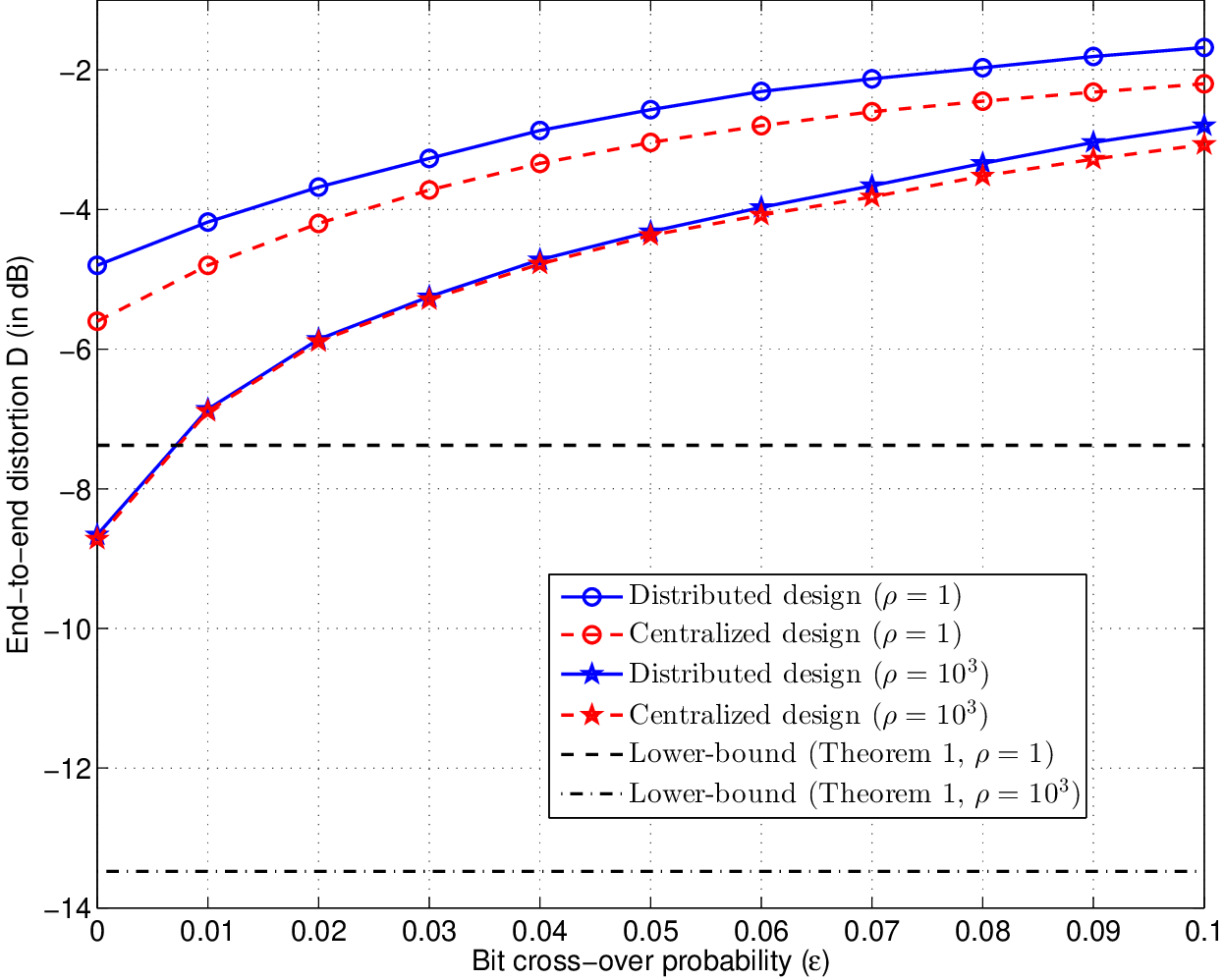}\\
  \caption{End-to-end MSE ($D$ in dB) vs. channel bit cross-over probability $\epsilon$ using the proposed design scheme along with the lower-bound \eqref{eq:total lb} for different values of correlation ratio $\rho$. The parameters are chosen as $N=10$, $K=2$ and $M = 5$ for clean CS measurements, and quantization rate is set to $R = R_1 + R_2 = 10$ bits/vector, with $R_1 = R_2$.}
  \label{fig:MSE_eps}
  \end{center}
\end{figure}

\section{Conclusions and Future Works} \label{sec:conclusion}

We have studied the design and analysis of distributed vector quantization for CS measurements of correlated sparse sources over noisy channels. Necessary conditions for optimality of VQ encoder-decoder pairs have been derived with respect to minimizing end-to-end MSE. We have analyzed the MSE and showed that, without loss of optimality, it is the sum of CS reconstruction MSE and quantized transmission MSE, and we used this fact to derive a lower-bound on the end-to-end MSE. Simulation results have revealed that correlation between sources is an effective factor on the performance in addition to compression resources such as measurement and quantization rates. Further, in noisy channel scenarios, the proposed distributed design method provides robustness against channel noise. In addition, the performance of the distributed design closely follows that of the centralized design.

Finally, we mention that the paper was concerned with full-search VQ schemes suffering from exponential complexity, and hence all experiments were executed with low dimensions. To overcome the complexity issue, a potential future direction is to design sub-optimal structured VQ schemes, such as multistage and tree-structured VQ's for CS in the distributed setup considered in the paper.

\appendices
\section{Proof of \proref{theo1}} \label{app A}

The MMSE estimator that minimizes $D_{cs}$ in \eqref{eq:CS distotion} (given noisy CS measurements $\mathbf{y} = [\mathbf{y}_1^\top \hspace{0.1cm} \mathbf{y}_2^\top]^\top$) is $\widetilde{\mathbf{x}}_l^\star(\mathbf{y}) \triangleq \mathbb{E}[\mathbf{X}_l | \mathbf{y}]$ ($l \in \{1,2\}$). Marginalizing over all supports in $\mathbf{\Omega}$, we have
\begin{equation} \label{eq:margin mmse tot}
    \widetilde{\mathbf{x}}_l^\star(\mathbf{y}) = \sum_{\mathcal{S}\subset \mathbf{\Omega}} p(\mathcal{S}|\mathbf{y}) \mathbb{E}[\mathbf{X}_l|\mathbf{y},\mathcal{S}].
\end{equation}

Then, we note the following linear relation
\begin{equation} \label{eq:big matrix}
 \left[
      \begin{array}{c}
        \mathbf{Y}_1 \\
         \mathbf{Y}_2 \\
         \mathbf{\Theta}_\mathcal{S} \\
         \mathbf{Z}_{1,\mathcal{S}} \\
         \mathbf{Z}_{2,\mathcal{S}}
      \end{array}
    \right] =
    \left[
      \begin{array}{c c c c c}
        \mathbf{\Phi}_{1,\mathcal{S}} & \mathbf{\Phi}_{1,\mathcal{S}} & \mathbf{0} & \mathbf{I} & \mathbf{0}\\
         \mathbf{\Phi}_{2,\mathcal{S}} & \mathbf{0} & \mathbf{\Phi}_{2,\mathcal{S}} & \mathbf{0} & \mathbf{I} \\
         \mathbf{I} & \mathbf{0} & \mathbf{0} & \mathbf{0} & \mathbf{0} \\
         \mathbf{0}& \mathbf{I} & \mathbf{0} & \mathbf{0} & \mathbf{0} \\
         \mathbf{0} & \mathbf{0} & \mathbf{I} & \mathbf{0} & \mathbf{0} \\
      \end{array}
    \right] \cdot
    \left[
      \begin{array}{c}
        \mathbf{\Theta}_\mathcal{S} \\
         \mathbf{Z}_{1,\mathcal{S}} \\
         \mathbf{Z}_{2,\mathcal{S}} \\
         \mathbf{W}_{1} \\
         \mathbf{W}_{2}
      \end{array}
    \right],
\end{equation}
where for an arbitrary vector or a matrix $\mathbf{A}$, the notation $\mathbf{A}_\mathcal{S}$ represents the elements of $\mathbf{A}$ indexed by the support $\mathcal{S}$. Recalling that $\mathbf{\Theta}_\mathcal{S}$, $\mathbf{Z}_{l,\mathcal{S}}$ and $\mathbf{W}_l$ are all independent Gaussian vectors, the vector on the left-hand-side of \eqref{eq:big matrix} is jointly Gaussian. Therefore, based on \cite[Theorem 10.2]{93:Kay}, we have
\begin{equation}
    \mathbb{E}\left[[\mathbf{\Theta}_\mathcal{S}^\top \hspace{0.1cm}        \mathbf{Z}_{1,\mathcal{S}}^\top \hspace{0.1cm}         \mathbf{Z}_{2,\mathcal{S}}^\top]^\top \hspace{0.1cm} \big| \mathbf{y} \right] = \mathbf{C}^\top \mathbf{D}^{-1} \mathbf{y},
\end{equation}
where $\mathbf{C}$ and $\mathbf{D}$ are covariance matrices specified in \eqref{eq:accessories1} and \eqref{eq:accessories2}, respectively. Now, since $\mathbb{E} [\mathbf{X}_{l,\mathcal{S}} | \mathbf{y}] = \mathbb{E}[\mathbf{\Theta}_\mathcal{S}|\mathbf{y}] + \mathbb{E}[\mathbf{Z}_{l,\mathcal{S}}|\mathbf{y}]$ ($l \in \{1,2\}$), it follows that within support set $\mathcal{S}$, we obtain
\begin{equation}
    \widetilde{\mathbf{x}}^\star(\mathbf{y},\mathcal{S}) \triangleq \left[\begin{array}{c}
    \widetilde{\mathbf{x}}_1^\star(\mathbf{y},\mathcal{S}) \\
    \widetilde{\mathbf{x}}_2^\star(\mathbf{y},\mathcal{S})\\
  \end{array}\right] =
     \left[\begin{array}{c  c  c}
    \mathbf{I} & \mathbf{I} & \mathbf{0}  \\
    \mathbf{I} & \mathbf{0} & \mathbf{I}\\
  \end{array}\right] \mathbf{C}^\top \mathbf{D}^{-1} \mathbf{y},
\end{equation}
and otherwise zero.

Now, it only remains to find an expression for $p(\mathcal{S} | \mathbf{y})$ in \eqref{eq:margin mmse tot}. Let us first define $\mathbf{q}_\mathcal{S} \triangleq [\mathbf{\theta}_\mathcal{S}^\top \hspace{0.1cm}  \mathbf{z}_{1,\mathcal{S}}^\top \hspace{0.1cm}  \mathbf{z}_{2,\mathcal{S}}^\top]^\top$, then
\begin{equation} \label{eq:bayes1}
\begin{aligned}
    p(\mathcal{S} | \mathbf{y}) &= \frac{p(\mathbf{y}| \mathcal{S}) p(\mathcal{S})}{ \sum_{\mathcal{S}} p(\mathbf{y}|\mathcal{S})p(\mathcal{S})}& \\
    &= \frac{\int_{\mathbf{q}_{\mathcal{S}}} p(\mathbf{q}_{\mathcal{S}}|\mathcal{S}) p(\mathbf{y}|\mathbf{q}_{\mathcal{S}},\mathcal{S}) d \mathbf{q}_{\mathcal{S}}}{\sum_{\mathcal{S}} \int_{\mathbf{q}_{\mathcal{S}}} p(\mathbf{q}_{\mathcal{S}}|\mathcal{S}) p(\mathbf{y}|\mathbf{q}_{\mathcal{S}},\mathcal{S}) d \mathbf{q}_{\mathcal{S}}},&
\end{aligned}
\end{equation}
where we used the fact that $p(\mathcal{S}) = \frac{1}{{N \choose K}}$. It can be verified that $p(\mathbf{q}_{\mathcal{S}}| \mathcal{S}) = \mathcal{N}(\mathbf{0}, \mathbf{E})$ and $p(\mathbf{y}|\mathbf{q}_{\mathcal{S}},\mathcal{S}) = \mathcal{N}(\mathbf{F q}_\mathcal{S},\mathbf{N})$, where the matrices $\mathbf{N}$, $\mathbf{E}$ and $\mathbf{F}$ are specified in \eqref{eq:accessories3}, \eqref{eq:accessories4} and \eqref{eq:accessories5}, respectively. Therefore, it follows that
\begin{equation} \label{eq:Gaussian distributions}
\begin{aligned}
    p(\mathbf{q}_{\mathcal{S}}|\mathcal{S}) &= \frac{1}{\sqrt{(2 \pi)^{3K} \det(\mathbf{E})}} e^{-\frac{1}{2} \mathbf{q}_\mathcal{S}^\top \mathbf{E}^{-1} \mathbf{q}_\mathcal{S}}&\\
    p(\mathbf{y}|\mathbf{q}_{\mathcal{S}},\mathcal{S}) &= \frac{1}{\sqrt{(2 \pi)^{2M} \det(\mathbf{N})}} e^{-\frac{1}{2} (\mathbf{y} - \mathbf{Fq}_\mathcal{S})^\top\mathbf{N}^{-1} (\mathbf{y} - \mathbf{Fq}_\mathcal{S})}.&
\end{aligned}
\end{equation}
Using the Gaussian distributions in \eqref{eq:Gaussian distributions}, we get
\begin{equation} \label{eq:integral_comp}
\begin{aligned}
    p(\mathbf{q}_{\mathcal{S}}|\mathcal{S}) p(\mathbf{y}|\mathbf{q}_{\mathcal{S}},\mathcal{S}) &\!=\! \frac{e^{- \mathbf{y}^\top\mathbf{N}^{-\!1} \mathbf{y}}}{\sqrt{(2 \pi)^{3K +\! 2M} \det(\mathbf{E}) \det(\mathbf{N})}}& \\
    &\times e^{-\frac{1}{2} \left(\mathbf{q}_\mathcal{S}^\top \left(\mathbf{E}^{-\!1} \!+\! \mathbf{F}^\top \mathbf{N}^{-\!1} \mathbf{F} \right)\mathbf{q}_\mathcal{S} \!-\! 2 \mathbf{y}^\top \mathbf{N}^{-\!1} \mathbf{F} \mathbf{q}_\mathcal{S}\right)}.&
\end{aligned}
\end{equation}
Now, using \cite[eq. (346)]{matrix_cookbook}, it yields
\begin{equation} \label{eq:Gaussian int}
\begin{aligned}
    &\int_{\mathbf{q}_{\mathcal{S}}} p(\mathbf{q}_{\mathcal{S}}|\mathcal{S}) p(\mathbf{y}|\mathbf{q}_{\mathcal{S}},\mathcal{S}) d \mathbf{q}_{\mathcal{S}} =&\\
    &\hspace{0.6cm}\frac{e^{\frac{1}{2}\mathbf{y}^\top \left( \mathbf{N}^{-1} \mathbf{F}^\top \left(\mathbf{E}^{-1} + \mathbf{F}^\top \mathbf{N}^{-1} \mathbf{F} \right)^{-1} \mathbf{F} \mathbf{N}^{-1} - \mathbf{N}^{-1}\right) \mathbf{y}}}{\sqrt{(2 \pi)^{2M-3K} \det(\mathbf{N}) \det(\mathbf{E}^{-1} + \mathbf{F}^\top \mathbf{N}^{-1} \mathbf{F})}}.&
\end{aligned}
\end{equation}
Plugging \eqref{eq:Gaussian int} back into \eqref{eq:bayes1} yields $\beta_\mathcal{S}$ in \eqref{eq:MMSE closed details2}, and the proof is concluded.  $\QED$

\section{Proof of \proref{pro:oracle bound}} \label{app B}

The oracle estimator is obtained by \eqref{eq:MMSE closed details1} given $\mathcal{S}^{(or)}$. Then, it follows from the law of total expectation that
\begin{equation} \label{eq:oracle mse}
    D_{cs}^{(or)} = \frac{1}{2K} \mathbb{E} \left[\sum_{l=1}^2 \mathbb{E}\left[\|\mathbf{X}_{l,\mathcal{S}^{(or)}} - \widetilde{\mathbf{X}}_{l,\mathcal{S}^{(or)}}\|_2^2 \right] \right],
\end{equation}
where the inner expectation is taken over the distribution of $\mathbf{X}_l$ given the oracle-known support, and the outer expectation is taken over all possibilities of oracle support set. Further, $\widetilde{\mathbf{X}}_{l,\mathcal{S}^{(or)}} \triangleq \mathbb{E}[\mathbf{X}_{l,\mathcal{S}^{(or)}}|\mathbf{Y}] \in \mathbb{R}^K$, $l \in \{1,2\}$.

Defining $\mathbf{Q}_\mathcal{S} \triangleq [\mathbf{\Theta}_\mathcal{S}^{\top} \hspace{0.1cm}  \mathbf{Z}_{1,\mathcal{S}}^{\top} \hspace{0.1cm}  \mathbf{Z}_{2,\mathcal{S}}^{\top}]^{\top}$, for $l\!=\!1$, we have
\begin{equation} \label{eq:cov MSE oracle1}
\begin{aligned}
    &\mathbb{E}\left[\left(\mathbf{X}_{1,\mathcal{S}^{(or)}} - \widetilde{\mathbf{X}}_{1,\mathcal{S}^{(or)}}\right) \left(\mathbf{X}_{1,\mathcal{S}^{(or)}} - \widetilde{\mathbf{X}}_{1,\mathcal{S}^{(or)}}\right)^\top \right]& \\
    &\!\stackrel{(a)}{=} \left[\mathbf{I} \hspace{0.15cm}  \mathbf{I} \hspace{0.15cm} \mathbf{0} \right] \mathbb{E}\left[\left(\mathbf{Q}_{\mathcal{S}^{(or)}} \!-\! \widetilde{\mathbf{Q}}_{\mathcal{S}^{(or)}}\right) \left(\mathbf{Q}_{\mathcal{S}^{(or)}} \!-\! \widetilde{\mathbf{Q}}_{\mathcal{S}^{(or)}}\right)^{\! \top} \right] \left[\mathbf{I} \hspace{0.15cm}  \mathbf{I} \hspace{0.15cm} \mathbf{0} \right]^\top& \\
    &\! \stackrel{(b)}{=} \left[\mathbf{I} \hspace{0.15cm}  \mathbf{I} \hspace{0.15cm} \mathbf{0} \right] \left(\mathbf{E} - \mathbf{C}^\top \mathbf{D}^{-1} \mathbf{C} \right) \left[\mathbf{I} \hspace{0.15cm}  \mathbf{I} \hspace{0.15cm} \mathbf{0} \right]^\top,&
\end{aligned}
\end{equation}
where $(a)$ follows from the fact that $\mathbf{X}_{1,\mathcal{S}} = \left[\mathbf{I} \hspace{0.15cm}  \mathbf{I} \hspace{0.15cm} \mathbf{0} \right] \mathbf{Q}_\mathcal{S}$, and $(b)$ can be shown from \cite[Theorem 10.2]{93:Kay}. Similarly, for $l=2$, we obtain
\begin{equation} \label{eq:cov MSE oracle2}
\begin{aligned}
    &\mathbb{E}\left[\left(\mathbf{X}_{2,\mathcal{S}^{(or)}} - \widetilde{\mathbf{X}}_{2,\mathcal{S}^{(or)}}\right) \left(\mathbf{X}_{2,\mathcal{S}^{(or)}} - \widetilde{\mathbf{X}}_{2,\mathcal{S}^{(or)}}\right)^\top \right]
     &\\
     &=\left[\mathbf{I} \hspace{0.15cm}  \mathbf{0} \hspace{0.15cm} \mathbf{I} \right] \left(\mathbf{E} - \mathbf{C}^\top \mathbf{D}^{-1} \mathbf{C} \right) \left[\mathbf{I} \hspace{0.15cm}  \mathbf{0} \hspace{0.15cm} \mathbf{I} \right]^\top.&
\end{aligned}
\end{equation}

Combining \eqref{eq:cov MSE oracle1} and \eqref{eq:cov MSE oracle2} with \eqref{eq:oracle mse}, it follows that
\begin{equation} \label{eq:final oracle proof}
\begin{aligned}
    D_{cs}^{(or)} & \stackrel{(a)}{=} \frac{1}{2K}\mathbb{E}\left[\text{Tr}\left\{
    \left[\begin{array}{c c c}
    2\mathbf{I} & \mathbf{I} & \mathbf{I} \\
    \mathbf{I} & \mathbf{I} & \mathbf{0}  \\
    \mathbf{I} & \mathbf{I} & \mathbf{I} \\
    \end{array}\right] \left(\mathbf{E} - \mathbf{C}^\top \mathbf{D}^{-1} \mathbf{C} \right) \right\} \right]& \\
    &\stackrel{(b)}{=} \! 1 \!-\! \frac{1}{2K}\left[\text{Tr}\left\{
    \left[\begin{array}{c c c}
    2\mathbf{I} & \mathbf{I} & \mathbf{I} \\
    \mathbf{I} & \mathbf{I} & \mathbf{0}  \\
    \mathbf{I} & \mathbf{I} & \mathbf{I} \\
    \end{array}\right]  \frac{1}{{N \choose K}}\sum_{\mathcal{S}^{(or)}} \mathbf{C}^{\!\top} \mathbf{D}^{-\!1} \mathbf{C}  \right\} \right]& \\
\end{aligned}
\end{equation}
where $(a)$ follows from fact that for two matrices $\mathbf{A}$ and $\mathbf{B}$ with appropriate dimensions, we have $\text{Tr}\{\mathbf{A + B}\} = \text{Tr}\{\mathbf{A}\} + \text{Tr}\{\mathbf{B}\}$ and $\text{Tr}\{\mathbf{AB}\} = \text{Tr}\{\mathbf{BA}\}$. Also, $(b)$ follows the uniform distribution of a possible oracle-known support set, and simple matrix algebra. $\QED$

\section{Proof of \proref{theo1-2}} \label{app C}

The proof follows the same line of arguments in the proof of \proref{theo1}. Let $\mathbf{q}_\mathcal{S} \triangleq [\mathbf{\theta_s}^\top \hspace{0.1cm}  \mathbf{z}_{1,\mathcal{S}}^\top \hspace{0.1cm}  \mathbf{z}_{2,\mathcal{S}}^\top]^\top$ and $\mathbf{g}_\mathcal{S} \triangleq [\mathbf{\theta_s}^\top \hspace{0.1cm}  \mathbf{z}_{1,\mathcal{S}}^\top]^\top$, then we rewrite $p(\mathbf{y}_2|\mathbf{y}_1)$ as
\begin{equation} \label{eq:proof cond1}
\begin{aligned}
    p(\mathbf{y}_2|\mathbf{y}_1) &= \frac{\sum_\mathcal{S} p(\mathbf{y}_1,\mathbf{y}_2 | \mathcal{S}) P(\mathcal{S})}{\sum_\mathcal{S} p(\mathbf{y}_1 | \mathcal{S}) P(\mathcal{S})}& \\
    &= \frac{\sum_\mathcal{S} \int_{\mathbf{q}_{\mathcal{S}}} p(\mathbf{q}_{\mathcal{S}}|\mathcal{S}) p(\mathbf{y}_1,\mathbf{y}_2|\mathbf{q}_{\mathcal{S}},\mathcal{S}) d \mathbf{q}_{\mathcal{S}}}{\sum_{\mathcal{S}} \int_{\mathbf{g}_{\mathcal{S}}} p(\mathbf{g}_{\mathcal{S}}|\mathcal{S}) p(\mathbf{y}_1|\mathbf{g}_{\mathcal{S}},\mathcal{S}) d \mathbf{g}_{\mathcal{S}}}.&
\end{aligned}
\end{equation}
It was shown in the proof of \proref{theo1} that $p(\mathbf{q}_\mathcal{S}|\mathcal{S})$ and $p(\mathbf{y}_1,\mathbf{y}_2 | \mathbf{q}_\mathcal{S},\mathcal{S})$ are Gaussian pdf's with known mean vectors and covariance matrices. Further, it can be easily shown that $p(\mathbf{g}_\mathcal{S}|\mathcal{S}) = \mathcal{N}(\mathbf{0},\text{bdiag}(\sigma_\theta^2 \mathbf{I},\sigma_z^2 \mathbf{I}))$, where $\text{bdiag}(\cdot)$ denotes the block diagonal matrix with diagonal matrices $\sigma_\theta^2 \mathbf{I}$ and $\sigma_z^2 \mathbf{I}$ as its blocks. Also, $p(\mathbf{y}_1|\mathbf{g}_\mathcal{S},\mathcal{S}) = \mathcal{N}(\mathbf{\Psi}\mathbf{g}_\mathcal{S},\sigma_{w_1}^2 \mathbf{I})$, where $\mathbf{\Psi} \triangleq [\mathbf{\Phi}_{1,\mathcal{S}} \hspace{0.15cm} \mathbf{\Phi}_{1,\mathcal{S}}]$. The integrations in the numerator and denominator of the last equation in \eqref{eq:proof cond1} can be analytically derived, similar to the ones in the proof of \proref{theo1}, leading to the expression in \eqref{cond prob}. $\QED$

\section{Proof of \theoref{theo2}} \label{app D}

We start with the decomposition of the end-to-end MSE in \eqref{eq:MSE decomp} as $D = D_{cs} + D_q$. Finding expressions for $D_q$ is non-trivial due to lack of analytical tractability and unknown probability distribution of sparse sources, and their MMSE reconstructions. Alternatively, we introduce two lower-bounds to $D$. The first relation is

\begin{equation} \label{eq:CS lb}
    D > D_{cs} \geq D_{cs}^{(or)}.
\end{equation}

Next, we note that the performance of the studied system is always poorer than that of a system where $\mathbf{X}_1$ and $\mathbf{X}_2$ are available for coding directly (with oracle known support set $\mathcal{S}^{(or)}$). Hence, we have
\begin{equation} \label{eq:quant dist der}
\begin{aligned}
    D \geq \frac{1}{2K} \sum_{l=1}^2 \mathbb{E}[\|\mathbf{X}_{l}|_{\mathcal{S}^{(or)}} - \widehat{\mathbf{X}}_{l}|_{\mathcal{S}^{(or)}}\|_2^2],
\end{aligned}
\end{equation}
where we denote by $\mathbf{X}_{l}|_{\mathcal{S}^{(or)}} \in \mathbb{R}^N$ the source $\mathbf{X}_{l}$ with oracle known support set $\mathcal{S}^{(or)}$, and $\widehat{\mathbf{X}}_{l}|_{\mathcal{S}^{(or)}} \in \mathbb{R}^N$ denotes decoded vector with known support. Since elements of support set are iid and uniformly drawn from all possibilities, a natural approach is to allocate $R_0 = \log_2 {N \choose K}$ bits to transmit $\mathcal{S}^{(or)}$ which is received without loss. Then, we only need to find the distortion-rate function for two correlated Gaussian sources using $R_1 + R_2 - \log_2 {N \choose K}$. Let us denote the non-sparse correlated Gaussian sources by $\mathbf{X}_{1,\mathcal{S}} , \mathbf{X}_{2,\mathcal{S}} \in \mathbb{R}^K$. The rate region for the quadratic Gaussian problem of two-terminal source coding has been developed in \cite{08:Wagner} so that we can lower-bound the last expression in \eqref{eq:quant dist der}. For this purpose, let us define $D_{l,\mathcal{S}^{(or)}} \triangleq \frac{1}{K}\mathbb{E}[\|\mathbf{X}_{l,\mathcal{S}^{(or)}} - \widehat{\mathbf{X}}_{l,\mathcal{S}^{(or)}}\|_2^2]$, $l \in \{1,2\}$, then with some mathematical simplifications of the results in \cite[Theorem 1]{08:Wagner}, we obtain
\begin{equation} \label{eq:quant dist Wagner}
\begin{aligned}
	D_{1,\mathcal{S}^{(or)}} D_{2,\mathcal{S}^{(or)}} &\geq  \left(1\!-\! \frac{\rho^2}{(1-\rho)^2}\right) 2^{\frac{-2\left(R_1 + R_2 -  \log_2 {N \choose K}\right)}{K}} &\\
& + \frac{\rho^2}{(1-\rho)^2} 2^{\frac{-4\left(R_1 + R_2 -  \log_2 {N \choose K}\right)}{K}}.&
\end{aligned}
\end{equation}
Since $D_{1,\mathcal{S}^{(or)}}$ is inversely proportional to $D_{2,\mathcal{S}^{(or)}}$, then $\frac{1}{2K}\sum_{l=1}^2 D_{l,\mathcal{S}^{(or)}}$ is minimized by setting $D_{1,\mathcal{S}^{(or)}} = D_{2,\mathcal{S}^{(or)}}$. Combining this fact with \eqref{eq:quant dist Wagner} and \eqref{eq:quant dist der}, it follows that
\begin{equation} \label{eq:quant dist total}
\begin{aligned}
	D  &\geq  \left(\left(1\!-\! \frac{\rho^2}{(1-\rho)^2}\right) 2^{\frac{-2\left(R_1 + R_2 -  \log_2 {N \choose K}\right)}{K}} \right.&\\
& \left. + \frac{\rho^2}{(1-\rho)^2} 2^{\frac{-4\left(R_1 + R_2 -  \log_2 {N \choose K}\right)}{K}} \right)^{\frac{1}{2}}
\triangleq D_q^{(or)}.&
\end{aligned}
\end{equation}

From the lower-bounds \eqref{eq:CS lb} and \eqref{eq:quant dist total}, it can be inferred that the former is tighter when CS measurements are noisy, and the latter is tighter when there is no loss due to CS distortion. Therefore, in order to adaptively consider both regimes, we develop a composite lower-bound by combining them as
\begin{equation}
    D > \max \left\{D_{cs}^{(or)},D_{q}^{(or)} \right\},
\end{equation}
which concludes the proof.

It can be also seen from \eqref{eq:total lb} that the channel aspects are not considered in developing the lower-bound. This is due to fact that the source-channel separation theorem is not optimal in the case of our studied distributed system, therefore, the minimum MSE (in terms of distortion-rate function over a DMC) cannot be analytically derived (based on channel capacity) in the scenario of noisy channels. As a result, when channel becomes very noisy, the lower-bound is not theoretically attainable. $\QED$

\bibliographystyle{IEEEtran}
\bibliography{IEEEfull,bibliokthPasha}
\end{document}